\documentclass[11pt,fleqn]{article}

\usepackage{amsmath,amssymb,hyperref,amsthm,graphics,epsfig,enumitem,cite}

\hypersetup{
    colorlinks=true,
    linkcolor={blue!90!black},    
    citecolor={blue!90!black},   
    urlcolor=blue      
}

\usepackage{tikz}
\usetikzlibrary{shapes.geometric, backgrounds, scopes}

\newcommand{\todo}[1][\null]{\ensuremath{\clubsuit}}
\newcommand{\checked}[1][\null]{\ensuremath{\diamond}}
\newcommand{\noprint}[1]{}

\newcommand{\Dx}{D_x}
\newcommand{\Dt}{D_t}
\DeclareMathOperator{\Eu}{E_u}

\newcounter{tbn}

\newcounter{mcasenum}

\newtheorem{theorem}{Theorem}

\newtheorem*{proposition*}{Proposition}
{\theoremstyle{definition}

\newtheorem{remark}{Remark}

}

\begin{document}\allowdisplaybreaks
\begin{center}
\par\noindent {\LARGE\bf
Extended group analysis and conservation laws of a class of variable coefficient generalized Kawahara equations
\par}

{\vspace{5mm}\par\noindent\large Olena~Vaneeva$^{1,2}$\footnote{Corresponding author. E-mail: vaneeva@imath.kiev.ua}, Alexander Zhalij$^{1}$, Olena Magda$^{2}$ and Oksana Brahinets$^{3}$
\par\vspace{1mm}\par}
\end{center}

{\par\noindent\it\small
${}^{1}$\ Institute of Mathematics of the National Academy of Sciences of Ukraine,\\[1ex]
$\phantom{{}^1}$\ 3 Tereshchenkivska Str., 01024 Kyiv-4, Ukraine\\[1ex]
${}^2$\ National Technical University of Ukraine “Igor Sikorsky Kyiv Polytechnic Institute”,\\[1ex]
$\phantom{{}^1}$\ 37 Beresteiskyi Ave., 03056 Kyiv, Ukraine  \\[1ex]
${}^3$\ Petro Mohyla Black Sea National University, 10, 68 Desantnykiv Str., 54003 Mykolaiv, Ukraine

}

{\vspace{5mm}\par\noindent\hspace*{5mm}\parbox{150mm}{\small
We review and extend the results on the group analysis of a class of generalized Kawahara equations with time-dependent coefficients. First, we provide an overview of the existing literature on Lie symmetries and Lie-invariant solutions of such equations. We then present a complete description of their transformation properties, including admissible, equivalence, and Lie symmetry transformations. For practical applications, we further extend these results by presenting a complete Lie symmetry classification without simplifying the coefficients via equivalence transformations. Lie reductions are then systematically performed, and several exact solutions are constructed. Low-order local conservation laws are exhaustively classified: every equation in this class admits conservation of mass and the squared $L^2$ norm, whereas energy-type conservation laws exist only for specific coefficient branches that align with the cases singled out by the symmetry classification. Finally, the classification results are enhanced by a study of contractions, which link cases of Lie symmetry extensions together with the associated reductions and conservation laws.
}\vspace{4mm}\par}

{\par\noindent\hspace*{5mm}\parbox{150mm}{\small
}\vspace{4mm}\par}

\section{Introduction}

Lie symmetry analysis provides an algorithmic framework for constructing reductions of nonlinear partial differential equations and, consequently, for obtaining exact solutions~\cite{Olver1986,Ovsiannikov1982,BlumanShevyakovAnco}. Moreover, Lie symmetries can be used as a selection principle for identifying physically relevant models within a broader class of admissible equations~\cite{FN}. The related investigation of admissible transformations, in particular equivalence transformations, makes it possible to apply equivalence-based methods to the construction of exact solutions for families of similar nonlinear partial differential equations. In this section, we present an overview of results on the study of Lie symmetries of Kawahara-type equations and their various generalizations.

The Kawahara equation
\begin{gather}\label{eqKawahara_classical}
u_t+\alpha u u_x +\beta u_{xxx}
+\gamma u_{xxxxx}=0,
\end{gather}
where $\alpha$, $\beta$, and~$\gamma$ are nonzero constants, is a classical model in solitary wave theory~\cite{Kawahara1972}.
In the usual sense, solitary waves are nonlinear waves of permanent form that decay rapidly (usually exponentially) in their tail regions. However, under critical conditions in dispersive systems (e.g., magneto-acoustic waves in plasmas or waves with surface tension), an unexpected emergence of weakly nonlocal solitary waves can occur. These waves consist of a central core similar to that of classical solitary waves, but they are accompanied by copropagating oscillatory tails that extend indefinitely far from the core with a nonzero constant amplitude.
To explain and describe the properties of these waves, a generalized nonlinear dispersive equation of the Korteweg--de Vries type with an additional fifth-derivative term, namely equation~\eqref{eqKawahara_classical}, was proposed~\cite{Hasimoto,Kawahara1972}.

Generalized models with constant coefficients related to the Kawahara equation appeared later (see, in particular,~\cite{Marchenko1988,Tkachenko&Yakovlev1999}).
For example, long waves in shallow water beneath an ice cover are modeled by the equation
\[u_t+u_x+\alpha u u_x +\beta u_{xxx}
+\gamma u_{xxxxx}=0.\]
This equation can be reduced to the classical Kawahara equation by simple changes of variables:
$\tilde x=x-t$ ($t$ and~$u$ are not transformed)
or $\tilde u=1+\alpha u$ ($t$ and~$x$ are not transformed).

To the best of our knowledge, the Lie symmetries of the Kawahara equation~\eqref{eqKawahara_classical} were first presented in~\cite{Hounkonnou} and Lie reductions are studied in~\cite{Hounkonnou2009}. The maximal Lie symmetry algebra of this equation is $\langle \partial_t, \partial_x, \alpha t\partial_x+\partial_u \rangle$.
These results were also obtained in~\cite{Liu2010f}, where the Lie symmetries of the modified Kawahara equation
\begin{gather}\label{eq_quadratic_Kawahara}
u_t+\alpha u^2 u_x +\beta u_{xxx}
+\gamma u_{xxxxx}=0,
\end{gather}
were found as well, where $\alpha$, $\beta$, and~$\gamma$ are nonzero constants.  The maximal Lie symmetry algebra of this equation is the two-dimensional Abelian algebra $\langle \partial_t, \partial_x\rangle$.
The results on the Lie symmetries and Lie reductions of  equations~\eqref{eqKawahara_classical} and~\eqref{eq_quadratic_Kawahara} appeared also in papers~\cite{BadaliHashemiGhahremani,KhorshidiGoodarzi} without citing previous works containing these symmetries.

In~\cite{Vasicek}, formal and generalized symmetries as well as local conservation laws of Kawahara equations with constant coefficients $\beta$, $\gamma$ and arbitrary nonlinearity $f(u)$,
\[u_t+f(u) u_x +\beta u_{xxx}+\gamma u_{xxxxx}=0,\]
where $f_u\beta\gamma\neq0$, were classified. In particular, it was proven that such equations admit only generalized symmetries which are equivalent to Lie point ones.

For more than a decade, the main attention of researchers has been focused on Kawahara-type models with coefficients depending on the time variable:
\begin{gather}\label{eq_ggKawahara}
u_t+\alpha(t)f(u)u_x+\beta(t)u_{xxx}+\gamma(t)u_{xxxxx}=0,\quad f_u\alpha\beta\gamma\neq0,
\end{gather}
where $f$ is a smooth function of~$u$ with $f_u\neq0$, and $\alpha$, $\beta$, and~$\gamma$ are smooth nonvanishing functions of~$t$, as well as on their various subclasses.

In~\cite{kaur2012a}, Lie symmetries were applied to find exact solutions of the classical and modified Kawahara equations with variable coefficients of the form~\eqref{eq_ggKawahara} with $f(u)=u$ and~$f(u)=u^2$.
Only partial results on Lie symmetries were obtained in that work, since equivalence transformations were not employed.

Due to the use of admissible transformations and the partition of the class into normalized subclasses,
an exhaustive group classification of the subclass of equations~\eqref{eq_ggKawahara} with $f(u)=u^n$, where $n$ is an arbitrary nonzero constant, was carried out in~\cite{KPV2014}.
In addition,~\cite{KPV2014} contains a complete classification of Lie reductions, as well as the construction of some exact solutions and simplest conservation laws for such equations.

The work~\cite{GandariasRosaRecioAnco} is devoted to the search for Lie symmetries and local conservation laws of generalized Kawahara equations of the form~\eqref{eq_ggKawahara}.
To simplify the problem, the arbitrary element was gauged to $\gamma=1$ in that work; however, this choice is not optimal.
In~\cite{VKS2016}, it was shown that the normalization property and the analysis of the type of equivalence group (usual / generalized / extended generalized) allow one to algorithmically choose an optimal gauging.
In~\cite{vane2020c,VMZh2023}, the complete group classification was obtained up to the respective equivalence groups using the gauge $\alpha=1$, and Lie reductions were performed.

In this paper we systematize and substantially extend the results on the
group analysis of the class~\eqref{eq_ggKawahara} obtained in our earlier
conference contributions~\cite{vane2020c,VMZh2023}.

We provide a self-contained and complete description of the
transformation properties of the class~\eqref{eq_ggKawahara}, including
a detailed proof of its partition into two normalized subclasses and
the construction of the respective extended generalized equivalence
groups $\hat{G}^\sim$ and $\hat{G}^\sim_0$ (Section~\ref{sec:admissible}).

The group classification results of~\cite{vane2020c} are enhanced by
identifying the types of the maximal Lie invariance algebras for all
cases in both classification tables (Tables~\ref{Kawahara_tab:1} and~\ref{Kawahara_tab:2}),
using the classification of low-dimensional Lie algebras~\cite{pate1977a}.

The group classification results are further extended by presenting
a complete Lie symmetry classification of the class~\eqref{eq_ggKawahara}
without simplifying the coefficients by equivalence transformations
(Table~\ref{Kawahara_tab:3} and Table~\ref{Kawahara_tab:4}), which is directly applicable to equations arising
in physical models.

The Lie reductions of~\cite{vane2020c,VMZh2023} are presented in a unified and
complete form covering all cases of both normalized subclasses, and the
list of exact solutions is enriched with new examples constructed using the equivalence method
(Section~\ref{sec:reductions}).

We classify the local conservation laws of low order for the equations
from the class~\eqref{eq_ggKawahara} (Section~\ref{sec:CL}). Every
equation of the class admits conservation laws of mass and of the squared $L^2$ norm, while energy-type conservation laws exist only for distinguished
coefficient branches, which agree with the cases singled out by the
symmetry classification. In particular, the conservation laws
constructed for the reciprocal-power, logarithmic and exponential nonlinearities
supplement the classification of~\cite{GandariasRosaRecioAnco}.

We carry out a study of contractions between cases of Lie symmetry
extensions (Section~\ref{sec:contractions}). Two types of contractions
are identified: Type~A connects the power nonlinearity cases ($f=u^n$)
to the exponential nonlinearity cases ($f=e^u$) within Table~\ref{Kawahara_tab:1},
and Type~B connects the logarithmic nonlinearity cases ($f=\ln u$) from
Table~\ref{Kawahara_tab:1} to the linear nonlinearity cases ($f=u+\kappa$) of
Table~\ref{Kawahara_tab:2}. For each contraction we verify that the equations,
algebra generators, ans\"atze, and reduced ODEs all transform consistently
in the limit.

\section{Admissible (form-preserving) transformations}\label{sec:admissible}

It is well known that no general theory exists for the integration of nonlinear partial differential equations (PDEs). Nevertheless, many special cases of complete integration or the construction of particular exact solutions are closely related to suitable changes of variables. Transformation-based methods, including the notable Lie symmetry method, constitute some of the most powerful analytical tools currently available for the study of nonlinear PDEs.

The systematic investigation of transformation properties of classes of nonlinear PDEs was initiated in 1991 by J.G. Kingston and C. Sophocleous~\cite{Kingston&Sophocleous1991}. Later these authors introduced the term form-preserving transformations to describe transformations that relate two particular equations within a given class while preserving the form of the equation and modifying only its arbitrary elements~\cite{Kingston&Sophocleous1998}. In 1992, J.P. Gazeau and P. Winternitz began studying similar transformations in classes of PDEs, referring to them as allowed transformations~\cite{Winternitz92}. Rigorous definitions and a systematic theoretical framework for these concepts were subsequently developed by R.O. Popovych~\cite{popo2006b,popo2010a}, who proposed the term admissible transformations and the related rigorous definition that unifies and formalizes the notion of form-preserving (allowed) transformations.

In brief, an admissible transformation is defined as a triple consisting of two fixed equations from a given class and a nondegenerate transformation that maps one equation to the other. The set of all admissible transformations of a class, equipped with the standard operation of composition, is referred to as the equivalence groupoid~\cite{popo2012a}. Equivalence groupoids play a fundamental role not only in group classification problems but also in a wide range of other studies related to classes of PDEs, including the construction of exact solutions and conservation laws, as well as the analysis of integrability~\cite{Popovych&Vaneeva2010,vane2014b}.
Equivalence transformations, which  play a fundamental role in the group analysis of differential equations, form a distinguished subset of admissible transformations. A crucial distinction is that admissible transformations do not, in general, possess a group structure, whereas equivalence transformations always constitute a group. An equivalence transformation maps any equation from a given class to another equation of the same class, while an admissible transformation may exist only for a specific pair of equations within the class under consideration.

According to L.V. Ovsiannikov, the equivalence group consists of nondegenerate point transformations acting on the independent and dependent variables as well as on the arbitrary elements of the class, with the transformations of the independent and dependent variables being projectable onto the corresponding space~\cite{Ovsiannikov1982}. With the subsequent introduction of other types of equivalence groups, the group possessing these properties is now referred to as the usual equivalence group. If the transformations of the independent and/or dependent variables involve arbitrary elements, the resulting structure is called a generalized equivalence group~\cite{Meleshko1994}. When new arbitrary elements depend nonlocally on the old ones, for example, through integral relations, the corresponding structure is termed an extended equivalence group~\cite{mogran}. A simultaneous relaxation of locality and projectability conditions leads to the notion of an extended generalized equivalence group. More recently, the concept of an effective generalized equivalence group was introduced in~\cite{Opanasenko2017}. This group is defined as a minimal subgroup of the full generalized equivalence group of a given class of PDEs that generates the same equivalence subgroupoid as the entire group.

A class of PDEs is said to be normalized in the usual sense if every admissible transformation within the class is induced by a transformation from its usual equivalence group. Analogous notions of normalization in the generalized, extended, and extended generalized senses are defined similarly~\cite{popo2010a}. If a class is normalized in the generalized sense, then its effective generalized equivalence group generates the entire equivalence groupoid of the class.
Once normalization of a given class is established, the determination of equivalence groupoids for its subclasses becomes significantly simpler, since they are necessarily subgroupoids of the equivalence groupoid of the original class. The rigorous theory of the equivalence groupoids has been developed in~\cite{vane2020b}.

The  study of the equivalence groupoid of the class~\eqref{eq_ggKawahara} using the direct method proposed in~\cite{Kingston&Sophocleous1998} shows that this class is not normalized. However, it is possible to partition it into two normalized subclasses singled out by the conditions $f_{uu}\neq0$ and $f_{uu}=0$.  Each of these two subclasses is normalized in the extended generalized sense, as their equivalence groups $\hat G^{\sim}$ and $\hat G_0^{\sim}$ depend on arbitrary element $\alpha$ and moreover depend on it nonlocally. The results are summarized in the following two statements.

\begin{theorem}\label{theorem_fuune0}
The subclass of the class~\eqref{eq_ggKawahara} singled out by the constraint \mbox{$f_{uu}\neq0$}  is normalized in the extended generalized sense. Its extended generalized equivalence group~$\hat G^{\sim}$ comprises the transformations
\begin{gather*}
\tilde t=T(t),\quad \tilde x=\delta_1(x+\delta_2 A) +\delta_3,\quad
\tilde u=\delta_4 u+\delta_5, \\
\tilde f=\delta_0 \left(f+\delta_2\right),\quad  \tilde\alpha=\dfrac{\delta_1}{\delta_0
T_t}\alpha, \quad \tilde A=\dfrac{\delta_1}{\delta_0}A+\varepsilon_0,\quad
 \tilde\beta=\dfrac{\delta_1{}^3}{T_t}\beta,\quad  \tilde\gamma=\dfrac{\delta_1{}^5}{T_t}\gamma,
\end{gather*}
where  $\varepsilon_0$ and $\delta_j,$ $j=0,1,\dots,5,$  are arbitrary constants with
$\delta_0\delta_1\delta_4\not=0$, and $T(t)$ is an arbitrary smooth function with $T_t\neq0.$
The additional arbitrary element $A$ sa\-tisfies the auxiliary equation  $A_t=\alpha$.
\end{theorem}
The  subclass of the class~\eqref{eq_ggKawahara}, singled out by the constraint \mbox{$f_{uu}=0$}, consists of the equations with $f=\kappa_1u+\kappa_0$, where $\kappa_1\neq0$. Since the constant $\kappa_1$ can be absorbed into $\alpha$ by the reparameterization $\tilde f=f/\kappa_1$, $\tilde\alpha=\kappa_1\alpha$, we set $\kappa_1=1$ without loss of generality, i.e. we consider the class of Kawahara equations of the form
\begin{equation}\label{eq_Kawahara_f=u}
u_t+\alpha(t)(u+\kappa_0)u_x+\beta(t)u_{xxx}+\gamma(t)u_{xxxxx}=0,
\end{equation}
where  $\alpha$, $\beta$, and~$\gamma$ are smooth nonzero functions of their arguments and $\kappa_0$ is an arbitrary constant. The following statement is true.
\begin{theorem}\label{theorem_fuu=0}
The class~\eqref{eq_Kawahara_f=u} is normalized in the extended generalized sense.
Its extended generalized equivalence group~$\hat G_0^{\sim}$ is constituted by the transformations of the form
\begin{gather*}
\tilde t=T(t),\quad
\tilde x=\dfrac{\varepsilon_2x+\varepsilon_1A+\varepsilon_0}{\delta_2A+\delta_1},\quad
\tilde u=\dfrac{\varepsilon_2}{\Delta}\left((\delta_2A+\delta_1)u-\delta_2x+\delta_2\kappa_0 A+\varepsilon_3\right),\\
\tilde A=\dfrac{\delta_2'A+\delta_1'}{\delta_2A+\delta_1},\quad\tilde\alpha=\dfrac{\Delta}{T_t(\delta_2A+\delta_1)^2}\alpha, \quad
\tilde \beta=\dfrac{\varepsilon_2{}^3}{T_t(\delta_2A+\delta_1)^3}\beta, \\
\tilde \gamma=\dfrac{\varepsilon_2{}^5}{T_t(\delta_2A+\delta_1)^5}\gamma,\quad \tilde \kappa_0=\dfrac1{\Delta}\left({\delta_1\varepsilon_2\kappa_0+\delta_1\varepsilon_1-\delta_2\varepsilon_0-\varepsilon_3\varepsilon_2}\right),
\end{gather*}
where
 $\delta_j, \delta_j'$ $j=1,2,$ and $\varepsilon_i$, $i=0,1,2,3,$ are arbitrary constants defined up to a nonzero multiplier, with $\Delta=\delta_2'\delta_1-\delta_1'\delta_2\neq0$ and
$\varepsilon_2\not=0$. The additional arbitrary element $A$ sa\-tisfies the auxiliary equation  $A_t=\alpha$.
\end{theorem}

The structure of the equivalence groupoid of the class~\eqref{eq_ggKawahara} is illustrated in Fig.~1.
\begin{figure}[htbp]
  \centering
 \begin{tikzpicture}[scale=0.8, transform shape]

    \tikzset{
      smooth curve/.style={
        thick,
        black
      },
      math label/.style={
        font=\Large\bfseries,
        execute at begin node=\boldmath
      }
    }

    \draw[smooth curve] (0, 0) ellipse (4.5 and 2.5);

    \node[math label] at (-2.3, 1.2) {$f_{uu} \neq 0$};
    \node[math label] at (-2.5, -0.4) {$\hat G^\sim$};

    \draw[smooth curve, fill=yellow!20, shift={(2, -0.2)}] (0, 0) ellipse (1.8 and 1.3);

    \node[math label] at (2.0, -0.6) {$\hat G_0^\sim$};
    \node[math label] at (2.0, 0.4) {$f_{uu} = 0$};

\end{tikzpicture}

   \caption{Partition of the class~\eqref{eq_ggKawahara} into the subclass $f_{uu}=0$ and its complement $f_{uu}\neq0$, together with the structure of the corresponding equivalence groupoid.}
  \label{fig:my_diagram}
\end{figure}
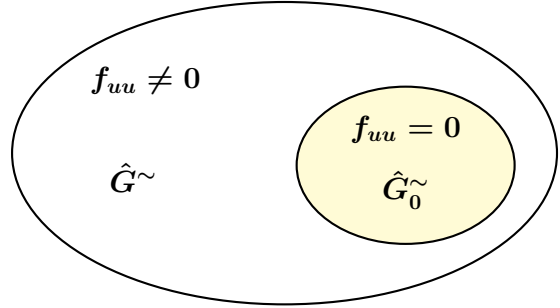
\begin{proof}[Proof of Theorems~\ref{theorem_fuune0} and~\ref{theorem_fuu=0}]
It was proved in~\cite{vane2014b} that the class
\begin{equation}\label{class_gen}
u_t=F(t)u_n+G(t,x,u,u_1,\dots,u_{n-1}),
\end{equation}
where $F\ne0,$ $G_{u_iu_{n-1}}=0,$ $i=1,\dots,n-1,$ $n\geqslant 2$,
$
u_{n}=\frac{\partial^n u}{\partial x^n},
$
$F$ and $G$ are arbitrary smooth functions of their variables, is normalized in the usual sense.  The transformation components for the variables $t,$
$x$ and $u$ of admissible transformations for the class~\eqref{class_gen} are of the form
\begin{eqnarray}\label{EqEquivtransOfvcKdVlikeSuperclass}
\tilde t=T(t),\quad \tilde x=X^1(t)x+X^0(t),\quad \tilde u=U^1(t,x)u+U^0(t,x),
\end{eqnarray}
where $T$, $X^1$, $X^0$, $U^1$ and $U^0$ are arbitrary smooth functions of their variables with $T_tX^1U^1\neq0$.

As class~\eqref{eq_ggKawahara} is a subclass of class~\eqref{class_gen} with $n=5,$ transformation components for the variables $t,$ $x,$ and $u$ of admissible transformations in~\eqref{eq_ggKawahara} can be sought  in the form~\eqref{EqEquivtransOfvcKdVlikeSuperclass}.
Following  the direct method~\cite{Kingston&Sophocleous1998}, we
suppose that equation~\eqref{eq_ggKawahara} is similar  to an equation from the same class,
\begin{equation}\label{eq_ggKawahara_tilda}
\tilde u_{\tilde t}+\tilde\alpha(\tilde t)\tilde f( \tilde u)\tilde u_{\tilde x}+\tilde\beta(\tilde t) \tilde u_{\tilde x\tilde x\tilde
x}+\tilde\gamma(\tilde t) \tilde u_{\tilde x\tilde x\tilde x\tilde x\tilde x}=0,
\end{equation}
and these two equations are connected by a nondegenerate point~transformation of the form~\eqref{EqEquivtransOfvcKdVlikeSuperclass}.
Rewriting~\eqref{eq_ggKawahara_tilda} in terms of the untilded variables,
we further substitute $u_t=-\alpha(t)f(u)u_x-\beta(t)u_{xxx}-\gamma(t)u_{xxxxx}$ to the derived equation.
Splitting the obtained identity with respect to the derivatives of $u$ leads to the determining
equations on the functions~$T$, $X^1$, $X^0$, $U^1$ and $U^0$, which result in the system
\begin{gather}\label{ddeteq1}
U^1_x=0,\quad \tilde\beta T_t-\beta(X^1)^3=0,\quad \tilde\gamma T_t-\gamma(X^1)^5=0,\\\label{ddeteq2}
U^0_x\alpha f+U^1_tu+\gamma U^0_{xxxxx}+\beta U^0_{xxx}+U^0_t=0,\\\label{ddeteq3}
\tilde\alpha\tilde f T_t=\alpha f X^1+X^1_tx+X^0_t.
\end{gather}

We proceed with the classification of all admissible transformations in the class~\eqref{eq_ggKawahara}.
Equation~\eqref{ddeteq2} implies that such transformations essentially  differ for the cases $f_{uu}\neq0$ and $f_{uu}=0$.

{\bf I.} If $f(u)$ is a nonlinear function, then equations \eqref{ddeteq2}, \eqref{ddeteq3} imply the conditions
\begin{equation*}
U^0_x=U^0_t=U^1_t=X^1_t=0, \quad \tilde\alpha\tilde f T_t=\alpha f X^1+X^0_t.
\end{equation*}
Solving these equations together with equations \eqref{ddeteq1} we find that $X^0_t$ is proportional to the arbitrary element $\alpha,$ so the explicit form of $x$-component of admissible transformations for the subclass of the class~\eqref{eq_ggKawahara} singled out by the condition $f_{uu}\neq0$ will be nonlocal with respect to the arbitrary element $\alpha$. In order for the equivalence transformations to remain point transformations and to be well defined, we extend the tuple of arbitrary elements $(f,\alpha,\beta,\gamma)$  with the additional arbitrary element $A$, that satisfies the auxiliary condition $A_t=\alpha.$
The solution of the resulting system gives the form of the transformations from the equivalence group $\hat G^{\sim}$ presented in Theorem~\ref{theorem_fuune0}.

{\bf II.} If $f(u)=\kappa_1 u+\kappa_0,$ where $\kappa_1$ is a nonzero constant and $\kappa_0$ is an arbitrary constant, then without loss of generality we can further consider the class~\eqref{eq_Kawahara_f=u},
where the real constant $\kappa_0$ and the nonvanishing smooth functions $\alpha(t),$ $\beta(t)$ and $\gamma(t)$ are arbitrary elements.
In this case splitting of equations \eqref{ddeteq2} and~\eqref{ddeteq3} with respect to $u$ results in the following conditions
\begin{gather*}
  U^0_x\alpha +U^1_t=0,\quad U^0_x\alpha \kappa_0 +U^0_t=0,\quad \tilde\alpha U^1T_t=\alpha X^1,\\
  \tilde\alpha T_t(U^0+\tilde \kappa_0)=\alpha \kappa_0 X^1+X^1_tx+X^0_t.
\end{gather*}
Solving these equations together with equations \eqref{ddeteq1}, we get the statement of Theorem~\ref{theorem_fuu=0}.
\end{proof}

\begin{theorem}
The usual equivalence group of the class~\eqref{eq_ggKawahara} consists of the transformations
\begin{gather*}
\tilde t=T(t),\quad \tilde x=\delta_1x +\delta_2,\quad
\tilde u=\delta_3 u+\delta_4, \\
\tilde f=\delta_0 f,\quad  \tilde\alpha=\dfrac{\delta_1}{\delta_0
T_t}\alpha, \quad
 \tilde\beta=\dfrac{\delta_1{}^3}{T_t}\beta,\quad  \tilde\gamma=\dfrac{\delta_1{}^5}{T_t}\gamma,
\end{gather*}
where  $\delta_j,$ $j=0,1,\dots,4,$  are arbitrary constants with
$\delta_0\delta_1\delta_3\not=0$, and $T(t)$ is an arbitrary smooth function with $T_t\neq0$.
\end{theorem}
\begin{remark}
 The group~$\hat G^{\sim}$ is also the extended generalized  equivalence group for the entire class~\eqref{eq_ggKawahara}, which is not normalized in contrast to its subclass with $f_{uu}\neq0.$ The effective generalized equivalence group~$\hat G^{\sim}$  for reparameterized class~\eqref{eq_ggKawahara} with the extended tuple of arbitrary elements $(f,\alpha,\beta,\gamma,A)$ is a nontrivial extended generalized equivalence group
for the initial class~\eqref{eq_ggKawahara}. Indeed, this group induces the maximal subgroupoid in the equivalence groupoid of the  class~\eqref{eq_ggKawahara} among all the equivalence groups of possible reparameterizations of this class. Moreover, the $x$-components of transformations from~$\hat G^{\sim}$ depend on the new arbitrary element~$A$.
\end{remark}

The equivalence groupoid of the subclass of the class~\eqref{eq_ggKawahara} singled out by the condition $f_{uu}\neq0$ (resp. of the class~\eqref{eq_Kawahara_f=u}) is generated
by its extended generalized equivalence group~$\hat G^{\sim}$ (resp. $\hat G_0^{\sim}$).
Therefore, the class~\eqref{eq_ggKawahara} is not normalized but it can be partitioned into two disjoint subclasses each of which is normalized. It can be shown that these groups are not only  effective generalized equivalence groups of the corresponding reparameterized classes but also their entire generalized equivalence groups (see the related definition in~\cite{Opanasenko2017}).

Using the results of Theorems~\ref{theorem_fuune0} and~\ref{theorem_fuu=0} we derive the criterion of reducibility of variable coefficient generalized Kawahara equations to their constant coefficient counterparts. The following statement is true.
\begin{theorem}
An equation from the class~\eqref{eq_ggKawahara} with variable coefficients $\alpha,$ $\beta$ and $\gamma$  is reducible to a constant coefficient equation from the
same class  if and only if the coefficients satisfy the conditions
\begin{gather*}\label{criterion1}
\left(\frac\beta\alpha\right)_t=0,\quad\left(\frac{\gamma\vphantom{\beta}}\alpha\right)_t=0, \quad\mbox{for}\quad f_{uu}\ne0,\\[1ex]
\label{criterion2}
\left(\frac1{\alpha}\left(\frac\beta\alpha\right)_t\right)_t=0,\quad
\left(\frac{\gamma\vphantom{\beta}\alpha^2}{\beta^3}\right)_t=0,\quad\mbox{for}\quad f_{uu}=0.
\end{gather*}
\end{theorem}

\section{Gauging of arbitrary elements}\label{sec:gauging}
The presence of the arbitrary function $T(t)$ in the equivalence
transformations from the group $\hat G^\sim$  of the entire class~\eqref{eq_ggKawahara} allows one to gauge
either $\alpha$ or $\beta$ or $\gamma$ to
a simple constant value, e.g., to 1.
An important question is which one of the three potential gaugings is the optimal one.
It was shown in~\cite{KPV2014,VKS2016} how to choose the optimal gauging using the normalization property of the class under consideration.
The classes  normalized in the usual sense are most convenient for investigation and the classes normalized in the extended generalized sense are the most complicated among normalized classes. If the class is not normalized, one should look for the possibility of a partition of such a class into normalized subclasses. In our case the class~\eqref{eq_ggKawahara} that is not normalized can be partitioned into two normalized subclasses singled out by the conditions
$f_{uu}\neq0$  and $f_{uu}=0$.  Gauging the arbitrary elements will be performed separately in each subclass.

{\bf I.} If $f_{uu}\neq0,$ then Theorem~\ref{theorem_fuune0} implies that the
subclasses of the  class~\eqref{eq_ggKawahara}  singled out by the conditions  $\beta=1$ or $\gamma=1$ will stay normalized only in the extended generalized sense, since equivalence transformations
for such subclasses will still involve $A$.
After gauging $\alpha=1$ in the subclass of~\eqref{eq_ggKawahara} with $f_{uu}\neq0$  we obtain the class normalized in the usual sense, as in this case up to gauging equivalence transformations we can set  $A=t$  and this nonlocality disappears.

The gauging $\alpha=1$ is realized by the family of point transformations from the equivalence group $\hat G^\sim$ with
$
\tilde t=\int^t_{t_0}\!\alpha(y)\,{\rm d}y,$ $\tilde x=x,$ $\tilde u=u,$
that maps equations from the
class~\eqref{eq_ggKawahara} with $f_{uu}\neq0$  to equations from the same class  with $\tilde
\alpha=1$, $\tilde\beta=\beta/\alpha$ and $\tilde\gamma=\gamma/\alpha$.
All results on symmetries, conservation laws, classical solutions and other related objects for equations~\eqref{eq_ggKawahara} with \mbox{$f_{uu}\neq0$} can be found
using the similar results derived for equations from its subclass
\begin{eqnarray}\label{eq:genKawahara}
u_t+f(u)u_x+\beta(t)u_{xxx}+\gamma(t)u_{xxxxx}=0,\quad f_{uu}\beta\gamma\neq0.
\end{eqnarray}
We derive the  equivalence groupoid of
the class~\eqref{eq:genKawahara} and formulate the following statement.
\begin{theorem}
The class~\eqref{eq:genKawahara}  is normalized in the usual sense. Its usual equivalence group~$G^{\sim}$ consists of the transformations
\begin{gather*}
\tilde t=\delta_1t+\delta_2,\quad \tilde x=\delta_3x+\delta_4 t +\delta_5,\quad
\tilde u=\delta_6 u+\delta_7, \\
\tilde f= \dfrac1{\delta_1} \left(\delta_3f+\delta_4\right),\quad
 \tilde\beta=\dfrac{\delta_3{}^3}{\delta_1}\beta,\quad  \tilde\gamma=\dfrac{\delta_3{}^5}{\delta_1}\gamma,
\end{gather*}
where  $\delta_j,$ $j=1,\dots,7,$  are arbitrary constants with
$\delta_1\delta_3\delta_6\not=0$.
\end{theorem}

{\bf II.} If $f_{uu}=0,$ then Theorem~\ref{theorem_fuu=0} implies that after the gauging $\alpha=1$ the respective subclass of the class~\eqref{eq_Kawahara_f=u} will be normalized in the generalized sense, since the dependence on the constant arbitrary element $\kappa_0$ will remain in the transformations.
So, the optimal gauging in this case is the simultaneous gauging of two arbitrary elements, $\alpha(t)=1$ and $\kappa_0=0.$
This gauging  is realized by the family of point transformations from the equivalence group $\hat G_0^\sim$ with
$
\tilde t=\int^t_{t_0}\!\alpha(y)\,{\rm d}y,$ $\tilde x=x,$ $\tilde u=u+\kappa_0,$
that maps equations from the
class~\eqref{eq_Kawahara_f=u}  to equations from the same class with $\tilde
\alpha=1$, $\tilde \kappa_0=0$, $\tilde\beta=\beta/\alpha$ and $\tilde\gamma=\gamma/\alpha$.
Without loss of generality we can restrict ourselves to the investigation of the following class
\begin{equation}\label{eq:genKawahara_u}
u_t+uu_x+\beta(t)u_{xxx}+\gamma(t)u_{xxxxx}=0,
\end{equation}
instead of its superclass~\eqref{eq_Kawahara_f=u}.

The study of the equivalence groupoid of the class~\eqref{eq:genKawahara_u} results in the following assertion.
\begin{theorem}
The class~\eqref{eq:genKawahara_u}   is normalized in the usual sense. Its usual equivalence group~$G^{\sim}_1$ consists of the transformations of the form
\begin{gather*}
\tilde t=\dfrac{\delta_4t+\delta_3}{\delta_2t+\delta_1},\quad  \tilde x=\dfrac{\varepsilon_2x+\varepsilon_1t+\varepsilon_0}{\delta_2t+\delta_1},
\quad  \tilde u=\dfrac{\varepsilon_2(\delta_2t+\delta_1)u-\varepsilon_2\delta_2x+\varepsilon_1\delta_1-\varepsilon_0\delta_2}{\Delta}, \\[0.5ex]
\tilde\beta=\dfrac{\varepsilon_2{}^3}{(\delta_2t+\delta_1)\Delta}\beta,\quad  \tilde\gamma=\dfrac{\varepsilon_2{}^5}{(\delta_2t+\delta_1)^3\Delta}\gamma,
\end{gather*}
where   $\varepsilon_i,$ $i=0,1,2,$  and $\delta_j,$ $j=1,\dots,4,$  are arbitrary constants defined up to a nonzero multiplier;
$\varepsilon_2\not=0$ and $\Delta=\delta_1\delta_4-\delta_2\delta_3\neq0.$
\end{theorem}

In the next section  we demonstrate that  the chosen gaugings allow us to solve exhaustively the group classification problems for both derived normalized subclasses of the class~\eqref{eq_ggKawahara}.

\section{Group classification}
 As in the previous section we consider separately two normalized subclasses of the class~\eqref{eq_ggKawahara} that are singled out by the conditions
 $f_{uu}\neq0$ and $f_{uu}=0.$

{\bf I.}
The group classification problem for the class~\eqref{eq_ggKawahara} with $f_{uu}\neq0$ up to $\hat G^{\sim}$-equivalence reduces to
the similar problem for the class~\eqref{eq:genKawahara} up to $G^{\sim}$-equivalence.
The group classification of the class~\eqref{eq:genKawahara} is performed
using the  classical algorithm based on direct integration of the determining equations implied by the infinitesimal invariance
criterion~\cite{Olver1986,Ovsiannikov1982}.
The symmetry generators of the form $Q=\tau(t,x,u)\partial_t+\xi(t,x,u)\partial_x+\eta(t,x,u)\partial_u$
must satisfy the criterion of infinitesimal invariance,
\begin{equation*}\label{c2}
Q^{(5)}\{u_t+f(u)u_x+\beta(t)u_{xxx}+\gamma(t)u_{xxxxx}\}\Big|_{u_t=-f(u)u_x-\beta(t)u_{xxx}-\gamma(t)u_{xxxxx}}=0,
\end{equation*} where  $Q^{(5)}$ is the fifth prolongation of the  vector field~$Q$~\cite{Olver1986,Ovsiannikov1982}.

The infinitesimal invariance criterion implies the determining equations, the simplest of which result in the following forms of $\tau,$ $\xi,$ and $\eta,$
\begin{equation*}
\tau=\tau(t),\quad
\xi=\xi^1(t)x+\xi^0(t), \quad
\eta=\left(2\xi^1(t)+\mu(t)\right)u+\eta^0(t,x),
\end{equation*}
where $\tau$, $\xi^1$, $\xi^0$, $\mu$ and $\eta^0$ are arbitrary smooth functions of their variables.

Then the rest of the determining equations have the form
\begin{gather}\label{deteq1}
\eta^0_xf+(2\xi^1_t+\mu_t)u+\eta^0_t+\eta^0_{xxx}\beta+\eta^0_{xxxxx}\gamma=0,\\\label{deteq2}
\left((2\xi^1+\mu)u+\eta^0\right)f_u=(\xi^1-\tau_t)f+\xi^1_tx+\xi^0_t,\\\label{deteq3}
\tau \gamma_t =(5\xi^1-\tau_t)\gamma,\quad \tau \beta_t =(3\xi^1-\tau_t)\beta.
\end{gather}

Below we give the sketch of the proof.
To find the kernel $A^\cap$  of the maximal Lie invariance algebras $\mathcal A^{\max}$ of equations  \eqref{eq:genKawahara} we split the determining equations with respect to the arbitrary elements and their derivatives, which
results in $\tau=\eta=0,$ $\xi=c_1,$ where $c_1$ is an arbitrary constant. Therefore,  $A^\cap$ is the one-dimensional algebra
$\langle\partial_x\rangle$ (Case 0 of Table~\ref{Kawahara_tab:1}).
At the next step we assume that $f$ is arbitrary and look for specifications of $\beta$ and $\gamma$ for which equations~\eqref{eq:genKawahara}
possess Lie symmetry extensions.
These are the cases $(\beta,\gamma)=(\lambda,\delta)$ and  $(\beta,\gamma)=(\lambda(pt+q)^2,\delta(pt+q)^4)$, where $\lambda,$ $\delta,$ $p$ and $q$
are constants and $\lambda\delta p\neq0.$ Using equivalence transformations from the group $G^\sim$ we  can set $\delta=\pm1,$
$p=1$ and $q=0.$ The respective
bases of the maximal Lie invariance algebras are adduced in Cases 1 and 2 of Table~\ref{Kawahara_tab:1}.

As $f_{uu}\neq0$, equation \eqref{deteq1} implies  $\eta^0_x=\eta^0_t=0$, so, $\eta^0=c_0$ and $\mu=-2\xi^1+c_1,$ where $c_0$ and $c_1$ are arbitrary constants.
Then equation~\eqref{deteq2} leads
to the condition
 $\xi^1_t=0$ and reduces to
\begin{equation}\label{deteq2a}
(c_1u+c_0)f_u=(\xi^1-\tau_t)f+\xi^0_t.
\end{equation}
The equations from the class~\eqref{eq:genKawahara}  possess  Lie symmetry extensions, when $f$ takes one of the following forms:
$f=\nu(p u+q)^n+r,$ $n\neq0,1,$ $f=\nu e^{p u}+r,$ and $f=\nu\ln(pu+q)+r,$ where $\nu$ and $p$ are nonzero constants, $q$ and $r$ are arbitrary constants.
Up to the $G^\sim$-equivalence this list is exhausted by the cases 1. $f=u^n,$ $n\neq0,1$, 2. $f=e^u$ and 3. $f=\ln u.$
Throughout the paper, all cases involving logarithms or noninteger powers of~$u$ are considered locally on domains where $u>0$; equivalently, $\ln u$ may be replaced by $\ln|u|$ on any fixed-sign domain. Similarly, expressions containing arbitrary real powers of~$t$ are understood locally on time intervals where they are real, smooth and nonvanishing; after a time translation, such an interval may be taken to lie in $t>0$.
Each of these forms of $f$ should be substituted into equation~\eqref{deteq2a}, then the final forms of the coefficients $\tau,$ $\xi,$ and $\eta$ are found and the classifying equations~\eqref{deteq3} give the possible forms of $\beta$ and $\gamma$ for which equations~\eqref{eq:genKawahara} possess Lie symmetry extensions.
The detailed proof for the case $f=u^n$, $n\neq0,1$, was given in~\cite{KPV2014}. The consideration of the other cases is analogous. For brevity we omit the detailed proof for these cases and formulate the final result in
the following statement.
\begin{theorem}
The kernel of the maximal Lie invariance algebras of equations from the class~\eqref{eq_ggKawahara}   with $f_{uu}\neq0$
coincides with the one-dimensional algebra $\langle\partial_x\rangle$.
All possible $\hat G^\sim$-inequivalent cases of extension of the maximal Lie invariance algebras
are exhausted by Cases 1--9 of Table~\ref{Kawahara_tab:1}.
\end{theorem}

\begin{table}[p!] \renewcommand{\arraystretch}{1.4}
\begin{center}
\setcounter{tbn}{-1}
\refstepcounter{table}\label{Kawahara_tab:1}
\textbf{Table~\thetable.}
The group classification of the class\\[1ex] $u_t+\alpha(t) f(u)u_x+\beta(t) u_{xxx}+\gamma(t) u_{xxxxx}=0$,\quad $f_{uu}\alpha\beta\gamma\neq0$.
\\[2ex]
\begin{tabular}{|c|c|c|c|l|c|}
\hline
no.&$\ f(u)\ $&$\ \beta(t)\ $&$\ \gamma(t)\ $&\hfil Basis of $\mathcal A^{\max}$&\hfil Type of $\mathcal A^{\max}$ \\
\hline
0&$\forall$&$\forall$&$\forall$&$\ \partial_x$&$A_1$\\
\hline
1&$\forall$&$\lambda t^2$&$\delta t^4$&
$\ \partial_x,\,t\partial_t+x\partial_x$&$A_2$\\
\hline
2&$\forall$&$\lambda $&$\delta$&
$\ \partial_x,\,\partial_t$&$2A_1$\\
\hline
3&$\ln u$&$\forall$&$\forall$&$\ \partial_x, t\partial_x+u\partial_u$&$2A_1$\\
\hline
4&$\ln u$&$\lambda t^2$&$\delta t^4$&$\ \partial_x, t\partial_x+u\partial_u, t\partial_t+x\partial_x$&$A_1\oplus A_2$\\
\hline
5&$\ln u$&$\lambda $&$\delta$&
$\ \partial_x,\, t\partial_x+u\partial_u,\,\partial_t$&$A_{3.1}$\\
\hline
6&$u^n$&
$\lambda t^\rho$&$\delta t^{\frac{5\rho+2}{3}}$&$\ \partial_x,\,3nt\partial_t+(\rho+1)n x\partial_x+(\rho-2) u\partial_u$&$\begin{array}{l}A_2,\ \rho\neq-1\\2A_1,\ \rho=-1\end{array}$\\
\hline
7&$u^n$&$\lambda e^{t}$&$\delta e^{\frac53 t}$&
$\ \partial_x,\,3n\partial_t+nx\partial_x+u\partial_u$&$A_2$\\
\hline
8&$e^u$&
$\lambda t^\rho$&$\ \delta t^{\frac{5\rho+2}{3}}\ $&$\ \partial_x,\,3t\partial_t+(\rho+1) x\partial_x+(\rho-2) \partial_u$&$\begin{array}{l}A_2,\ \rho\neq-1\\2A_1,\ \rho=-1\end{array}$\\
\hline
9&$e^u$&$\lambda e^{t}$&$\delta e^{\frac53 t}$&
$\ \partial_x,\,3\partial_t+x\partial_x+\partial_u$&$A_2$\\
\hline
\end{tabular}\\[2ex]
\parbox{145mm}{Here $\alpha(t)=1\bmod\, \hat G^\sim$, $\rho$ and $n$ are arbitrary constants, $n\neq0,1$; $\delta$ and $\lambda$ are nonzero constants, $\delta=\pm1\bmod\, \hat G^\sim$. }
\end{center}
\vspace{1ex}
\begin{center}
\setcounter{tbn}{-1}
\refstepcounter{table}\label{Kawahara_tab:2}
\textbf{Table~\thetable.}
The group classification of the class\\[1ex] $u_t+\alpha(t) (u+\kappa_0)u_x+\beta(t) u_{xxx}+\gamma(t) u_{xxxxx}=0$,\quad $\alpha\beta\gamma\neq0$.
\\[2ex]
\setlength{\tabcolsep}{2pt}
\begin{tabular}{|c|c|c|l|c|}
\hline
no.&$\beta(t)$&$\gamma(t)$&\hfil Basis of $\mathcal A^{\max}$&\hfil Type of $\mathcal A^{\max}$ \\
\hline
0&$\forall $&$\forall$&
$\ \partial_x,\,t\partial_x+\partial_u$&$2A_1$\\
\hline
1&
$\lambda t^\rho$&$\delta t^{\frac{5\rho+2}{3}}$&$\begin{array}{l} \partial_x,\,t\partial_x+\partial_u,\\3t\partial_t+(\rho+1) x\partial_x+(\rho-2) u\partial_u\end{array}$&
$\begin{array}{l}
A_1\oplus A_2,\ \rho\in\{-1,2\}\\[0.5ex]
A^a_{3.5},\ a=\tfrac{\rho+1}{\rho-2},\ \rho<\tfrac12, \rho\neq-1\\[0.5ex]
A^a_{3.5},\ a=\tfrac{\rho-2}{\rho+1},\ \rho>\tfrac12, \rho\neq2\\[0.5ex]
A_{3.4},\ \rho=\tfrac12
\end{array}$\\
\hline
2&$\lambda e^{t}$&$\delta e^{\frac53 t}$&
$\ \partial_x,\,t\partial_x+\partial_u,\,3\partial_t+x\partial_x+u\partial_u$&$A_{3.2}$\\
\hline
3&$\lambda $&$\delta$&
$\ \partial_x,\,t\partial_x+\partial_u,\,\partial_t$&$A_{3.1}$\\
\hline
4&\ \rule[-3.5ex]{0ex}{8ex}$\begin{array}{@{}c@{}}\lambda{(t^2\!+\!1)^\frac12}\times\\ e^{3\nu\arctan t}\end{array}$&$\ \begin{array}{@{}c@{}}\delta {(t^2\!+\!1)^\frac32}\times\\ e^{5\nu\arctan t}\end{array}\ $&
\ \parbox{60mm}{$\partial_x,\,t\partial_x+\partial_u,$\\[1ex] $(t^2+1)\partial_t+(t+\nu)x\partial_x+(x-(t-\nu)u)\partial_u$}&$ \begin{array}{@{}c@{}}A^{a}_{3.7},\ a=|\nu|, \nu\neq0\\A_{3.6},\  \nu=0\end{array}$\\
\hline
\end{tabular}
\\[2ex]
\parbox{150mm}{Here $\alpha=1\bmod\, \hat G_0^\sim$, $\kappa_0=0\bmod\, \hat G_0^\sim$, and $\rho$, $\nu$ are arbitrary constants. Up to $G_1^\sim$-equivalence one may impose $\rho\geqslant1/2$ and $\nu\geqslant0$. The constants $\delta$ and $\lambda$ are nonzero, with $\delta=\pm1\bmod\,G_1^\sim$.}
\end{center}
\end{table}
{\bf II.}
In the previous section we have shown that the group classification problem for the class~\eqref{eq_Kawahara_f=u} up to $\hat G_0^{\sim}$-equivalence reduces to
such a problem for the class~\eqref{eq:genKawahara_u} up to $G_1^{\sim}$-equivalence. The group classification of the  class~\eqref{eq:genKawahara_u} up to $G_1^{\sim}$-equivalence was carried out exhaustively in~\cite{KPV2014}. So we use the results derived therein and formulate the following statement.

 \begin{theorem}
The kernel of the maximal Lie invariance algebras of equations from the class~\eqref{eq_Kawahara_f=u}
coincides with the two-dimensional algebra $\langle\partial_x,\,T(t)\partial_x+\partial_u\rangle$, where $T_t=\alpha$.
Up to $\hat G_0^\sim$-equivalence, all possible cases of extension of the maximal Lie invariance algebras are represented by Cases~1--4 of Table~\ref{Kawahara_tab:2}; in the gauge $\alpha=1$  the kernel generator $T(t)\partial_x+\partial_u$ becomes $t\partial_x+\partial_u$.
\end{theorem}

 \begin{table}[p!]\small \renewcommand{\arraystretch}{1.55}
\begin{center}
 \setcounter{tbn}{-1}
\refstepcounter{table}\label{Kawahara_tab:3}
\textbf{Table~\thetable.}
The group classification of the class~$u_t+\alpha f(u)u_x+\beta u_{xxx}+\gamma u_{xxxxx}=0$,\quad $f_{uu}\alpha\beta\gamma\neq0$, \\without simplification by equivalence transformations.
\\[2ex]
\begin{tabular}{|c|c|c|c|l|}
\hline
no.&$f(u)$&$\beta(t)$&$\gamma(t)$&\hfil Basis of $\mathcal A^{\max}$ \\
\hline
0&$\forall$&$\forall$&$\forall$&$\partial_x$\\
\hline
1&$\forall$&$\lambda_1\alpha(T\!+\!s)^2 $&$\lambda_2\alpha(T\!+\!s)^4$&
$\partial_x,\,(T\!+\!s)\alpha^{-1}\partial_{t}+x\partial_x$\\
\hline
2&$\forall$&$\lambda_1\alpha $&$\lambda_2\alpha$&
$\partial_x,\,\alpha^{-1}\partial_{t}$\\
\hline
3&$\ln(pu\!+\!q)$&$\forall$&$\forall$&
$\partial_x,\,T\partial_x+\left(u+\tfrac qp\right)\!\partial_u$\\
\hline
4&$\ln(pu\!+\!q)$&$\lambda_1\alpha(T\!+\!s)^2$&$\lambda_2\alpha(T\!+\!s)^4$&
$\partial_x,\,T\partial_x+\left(u+\tfrac qp\right)\!\partial_u,\,(T\!+\!s)\alpha^{-1}\partial_{t}+x\partial_x$\\
\hline
5&$\ln(pu\!+\!q)$&$\lambda_1\alpha$&$\lambda_2\alpha$&
$\partial_x,\,T\partial_x+\left(u+\tfrac qp\right)\!\partial_u,\,\alpha^{-1}\partial_{t}$\\
\hline
6&$(pu\!+\!q)^n\!+\!r$&
$\lambda_1\alpha(T\!+\!s)^\rho$&$\lambda_2\alpha(T\!+\!s)^{\frac{5\rho+2}{3}}$&
$\partial_x,\,3n(T\!+\!s)\alpha^{-1}\partial_{t}+
n\!\left((\rho\!+\!1)x-(\rho\!-\!2)r(T\!+\!s)\right)\!\partial_{x}+$\\
&&&&$(\rho\!-\!2)\!\left(u+\tfrac qp\right)\!\partial_{u}$\\
\hline
7&$(pu\!+\!q)^n\!+\!r$&$\lambda_1\alpha e^{mT}$&$\lambda_2\alpha e^{\frac53 mT}$&
$\partial_x,\, 3n\alpha^{-1}\partial_{t}+nm(x-rT)\partial_{x}+
  m\!\left(u+\tfrac qp\right)\!\partial_{u}$\\
\hline
8&$e^{pu}\!+\!r$&
$\lambda_1\alpha(T\!+\!s)^\rho$&$\ \lambda_2\alpha(T\!+\!s)^{\frac{5\rho+2}{3}}\ $&
$\partial_x,\,$\\&&&&$3(T\!+\!s)\alpha^{-1}\partial_{t}+
\left((\rho\!+\!1)x-(\rho\!-\!2)r(T\!+\!s)\right)\!\partial_{x}+\tfrac{\rho-2}{p}\,\partial_{u}$\\
\hline
9&$e^{pu}\!+\!r$&$\lambda_1\alpha e^{mT}$&$\lambda_2\alpha e^{\frac53 mT}$&
$\partial_x,\,3\alpha^{-1}\partial_{t}+m(x-rT)\partial_{x}+\tfrac mp\,\partial_{u}$\\
\hline
\end{tabular}
\\[1.5ex]
\parbox{155mm}{Here $\lambda_1$, $\lambda_2$, $p$, $q$, $r$, $s$, $m$, $n$ and $\rho$ are arbitrary constants, $\lambda_1\lambda_2pm\ne0$, $n\neq0,1$; $\alpha$ is an arbitrary nonvanishing smooth function of $t$, $T=\int\!\alpha(t)\,{\rm d}t$.  }
\end{center}
\vspace{1ex}
\begin{center}
\setcounter{tbn}{-1}
\refstepcounter{table}\label{Kawahara_tab:4}
\textbf{Table~\thetable.}
The group classification of the class~$u_t+\alpha (u+\kappa_0)u_x+\beta u_{xxx}+\gamma u_{xxxxx}=0$,\quad $\alpha\beta\gamma\neq0$, \\without simplification by equivalence transformations.
\\[2ex]
\begin{tabular}{|@{\,}c@{\,}|@{\,}c@{\,}|@{\,}c@{\,}|@{\,\,}l@{\,}|}
\hline
no.&$\beta(t)$&$\gamma(t)$&\hfil Basis of $\mathcal A^{\max}$ \\
\hline
$0$&$\forall $&$\forall$&
$\partial_x,\,T\partial_x+\partial_u$\\
\hline
$1$&
$\lambda_1\alpha(aT\!+\!b)^\rho(cT\!+\!d)^{1-\rho}$&$\lambda_2\alpha(aT\!+\!b)^{\frac{5\rho+2}{3}}(cT\!+\!d)^{\frac{7-5\rho}{3}}$&
$\partial_x,\,T\partial_x+\partial_u, 3(aT\!+\!b)(cT\!+\!d)\alpha^{-1}\partial_{t}+
  $\\
  &&&$\left(3acT\!+\!ad(\rho\!+\!1)+bc(2\!-\!\rho)\!\right)x\partial_{x}+$\\
   &&&$\left(3acx\!-\!(3acT\!+\!ad(2\!-\!\rho)\!+\!bc(\rho\!+\!1)) (u+\kappa_0)\!\right)\partial_{u}$\\
\hline
$2$&$\lambda_1\alpha(cT\!+\!d)\exp\!\left(\frac{aT+b}{cT+d}\right)$&$\lambda_2\alpha(cT\!+\!d)^3\exp\!\left(\frac53 \frac{aT+b}{cT+d}\right)$&
$\partial_x,\,T\partial_x+\partial_u,\,3(cT\!+\!d)^2\alpha^{-1}\partial_{t}+$\\
 &&&$\left(3c(cT\!+\!d)\!+\!\Delta\right)x\partial_{x}+$\\
  &&&$
 \left(3c^2x\!+\! (\Delta\!-\!3c(cT\!+\!d))(u+\kappa_0)\right)\partial_{u}$\\
\hline
$3$&$\lambda_1\alpha(cT\!+\!d)$&$\lambda_2\alpha(cT\!+\!d)^3$&
$\partial_x,\,T\partial_x+\partial_u,\,(cT\!+\!d)^2\alpha^{-1}\partial_{t}+$\\
  &&&$c(cT\!+\!d)x\partial_{x}+
  c(c x\!-\!(cT\!+\!d) (u+\kappa_0))\partial_{ u}$\\
\hline
&$\lambda_1\alpha\exp\!\left(3\nu\arctan\frac{aT+b}{cT+d}\right)\!\!\times$&$\lambda_2\alpha\exp\!\left(5\nu\arctan\frac{aT+b}{cT+d}\right)\!\!\times$&
$\partial_x,\,T\partial_x+\partial_u,\,\left(\!(aT\!+\!b)^2\!+\!(cT\!+\!d)^2\!\right)\!\alpha^{-1}\partial_{t}+
$\\
$4$&$
\left((aT\!+\!b)^2\!+\!(cT\!+\!d)^2\right)^{\frac12}$&$\left(\!(aT\!+\!b)^2\!+\!(cT\!+\!d)^2\right)^{\frac32}$&
$  \left(a(aT\!+\!b)\!+\!c(cT\!+\!d)\!+\!\Delta\nu\right)x\partial_{x}+
  $\\
  &&&$\left(\!(a^2\!+\!c^2) x\!-\!(a(aT\!+\!b)\!+\!c(cT\!+\!d)\!-\!\Delta\nu) (u+\kappa_0)\!\right)\!\partial_{u}$\\
\hline
\end{tabular}
\\[1.5ex]
\parbox{155mm}{Here  $\lambda_1$, $\lambda_2$, $a$, $b$, $c$, $d$, $\rho$ and $\nu$ are arbitrary constants, $\lambda_1\lambda_2(c^2+d^2)\ne0$, and $\Delta=ad-bc\neq0$ in Cases 1, 2 and 4; $\alpha$ is an arbitrary nonvanishing smooth function of $t$, $T=\int\! \alpha(t) {\rm d}t$. The coefficient $\kappa_1$ of $u$ in $f=\kappa_1u+\kappa_0$ is set to 1 without loss of generality, since it can be absorbed into $\alpha$.}
\end{center}
\end{table}

The presented group classification reveals equations of the form~\eqref{eq_ggKawahara} that are of more potential interest for applications and  for which the
classical  Lie reduction method can be  used. The complete result was achieved using the knowledge of transformation properties of the class namely due to partition of the class into two normalized subclasses and choosing the optimal gauging for each of these subclasses.

To derive the complete list of Lie symmetry extensions for the entire class~\eqref{eq_ggKawahara},
where arbitrary elements are not simplified by point transformations, we use the equivalence-based approach~\cite{vane2012a}.
The results are collected in Tables~\ref{Kawahara_tab:3} and~\ref{Kawahara_tab:4}.

\section{Lie reductions and exact solutions}\label{sec:reductions}

The reduction method
with respect to subalgebras of Lie invariance algebras is algorithmic  and
well-known~\cite{Olver1986}.
As equations~\eqref{eq:genKawahara} are (1+1)-dimensional nonlinear partial differential equations, their Lie reductions with respect to one-dimensional subalgebras of the maximal Lie invariance algebras lead to ordinary differential equations (ODEs).
In order to  get inequivalent
reductions one should use subalgebras from the so called optimal system (see Section~3.3 in~\cite{Olver1986}).

 Consider the general form of Lie symmetry operator $Q=\tau(t,x,u)\partial_t+\xi(t,x,u)\partial_x+\eta(t,x,u)\partial_u$, which forms a basis of the respective one-dimensional subalgebra from the constructed optimal system. The ansatz is then found as a solution of the invariant surface condition $Q[u]:=\tau u_t+\xi u_x-\eta=0$. In practice,  the corresponding characteristic system $\frac{{\rm d}t}{\tau}=\frac{{\rm d}x}{\xi}=\frac{{\rm d}u}{\eta}$ should be solved.

The kernel of the maximal Lie invariance algebras of equations from the class~\eqref{eq:genKawahara}   with $f_{uu}\neq0$
coincides with the one-dimensional algebra $\langle\partial_x\rangle$.
For each Case~$N$ of Table~\ref{Kawahara_tab:1} (the case numbering coincides with that of Table~\ref{Kawahara_tab:1}, and $L_N$ denotes the corresponding equation with the gauge $\alpha=1$) we construct the optimal system of one-dimensional subalgebras of $\mathcal A^{\max}$ following Section~3.3 in~\cite{Olver1986} and carry out the Lie reductions with respect to its elements (see also~\cite{VMZh2023}). The reductions with respect to the subalgebra $\langle\partial_x\rangle$, which belongs to every optimal system, lead to constant solutions only and are omitted. The results are collected in Table~\ref{Kawahara_tab:5}, where for each case we list the remaining subalgebras of the optimal system, the associated ans\"atze, the invariant variables $\omega$ and the reduced equations $RL$. For Cases~6 and~8 the subcases $\rho\neq-1$ and $\rho=-1$ are distinguished as $6.1$, $6.2$ and $8.1$, $8.2$, respectively, since the corresponding maximal Lie invariance algebras are of different structure (cf.\ Table~\ref{Kawahara_tab:1}); in the subcases $N.2$ the coefficients are $\beta=\lambda t^{-1}$ and $\gamma=\delta t^{-1}$. Below $a$ and $\rho$ are arbitrary constants, $\lambda$ and $\delta$ are nonzero constants, $\varepsilon\in\{-1,0,1\}$, and the prime denotes differentiation with respect to~$\omega$.

\begin{table}[t!]\small \renewcommand{\arraystretch}{1.8}
\begin{center}
\setcounter{tbn}{-1}
\refstepcounter{table}\label{Kawahara_tab:5}
\textbf{Table~\thetable.}
Lie reductions of the equations from the class~\eqref{eq:genKawahara}.
\\[2ex]
\setlength{\tabcolsep}{2pt}
\begin{tabular}{|c|l|c|c|l|}
\hline
no.&\hfil Subalgebra&\hfil Ansatz, $u=$&$\omega$&\hfil Reduced equation \\
\hline
$1$&$t\partial_t+x\partial_x$&$ \varphi(\omega)$&$\tfrac xt$&
$\delta\varphi'''''+\lambda\varphi'''+(f(\varphi)-\omega)\varphi'=0$\\
\hline
$2$&$\partial_t+a\partial_x$&$ \varphi(\omega)$&$x\!-\!at$&
$\delta\varphi'''''+\lambda\varphi'''+(f(\varphi)-a)\varphi'=0$\\
\hline
$3$&$(t\!+\!a)\partial_x+u\partial_u$&$ e^{\frac x{t+a}}\varphi(\omega)$&$t$&
$\varphi'+\tfrac{\varphi\ln\varphi}{\omega+a}+\tfrac{\beta(\omega)\,\varphi}{(\omega+a)^3}+\tfrac{\gamma(\omega)\,\varphi}{(\omega+a)^5}=0$\\
\hline
$4_a$&$t\partial_t+(x\!+\!at)\partial_x+au\partial_u$&$ t^{a}\varphi(\omega)$&$\tfrac xt\!-\!a\ln t$&
$\delta\varphi'''''+\lambda\varphi'''+(\ln\varphi-\omega-a)\varphi'+a\varphi=0$\\
\hline
$4_\varepsilon$&$(t\!+\!\varepsilon)\partial_x+u\partial_u$&$ e^{\frac x{t+\varepsilon}}\varphi(\omega)$&$t$&
$\varphi'+\tfrac{\varphi\ln\varphi}{\omega+\varepsilon}+\tfrac{\lambda\omega^2\varphi}{(\omega+\varepsilon)^3}+\tfrac{\delta\omega^4\varphi}{(\omega+\varepsilon)^5}=0$\\
\hline
$5$&$\partial_t$&$ \varphi(\omega)$&$x$&
$\delta\varphi'''''+\lambda\varphi'''+\ln\varphi\,\varphi'=0$\\
\hline
$5_a$&$a\partial_t+t\partial_x+u\partial_u,\ a\neq0$&$ e^{\frac ta}\varphi(\omega)$&$x\!-\!\tfrac{t^2}{2a}$&
$\delta\varphi'''''+\lambda\varphi'''+\ln\varphi\,\varphi'+\tfrac1a\varphi=0$\\
\hline
$5_0$&$t\partial_x+u\partial_u$&$ e^{\frac xt}\varphi(\omega)$&$t$&
$\varphi'+\tfrac{\varphi\ln\varphi}{\omega}+\tfrac{\lambda\varphi}{\omega^3}+\tfrac{\delta\varphi}{\omega^5}=0$\\
\hline
$6.1$&$3nt\partial_t+(\rho\!+\!1)nx\partial_x+(\rho\!-\!2)u\partial_u$&$ t^{\frac{\rho-2}{3n}}\varphi(\omega)$&$xt^{-\frac{\rho+1}3}$&
$\delta\varphi'''''+\lambda\varphi'''+\varphi^n\varphi'-\tfrac{\rho+1}3\omega\varphi'+\tfrac{\rho-2}{3n}\varphi=0$\\
\hline
$6.2$&$nt\partial_t+a\partial_x-u\partial_u$&$ t^{-\frac1n}\varphi(\omega)$&$x\!-\!\tfrac an\ln t$&
$\delta\varphi'''''+\lambda\varphi'''+\varphi^n\varphi'-\tfrac an\varphi'-\tfrac1n\varphi=0$\\
\hline
$7$&$3n\partial_t+nx\partial_x+u\partial_u$&$ e^{\frac t{3n}}\varphi(\omega)$&$xe^{-\frac t3}$&
$\delta\varphi'''''+\lambda\varphi'''+\varphi^n\varphi'-\tfrac13\omega\varphi'+\tfrac1{3n}\varphi=0$\\
\hline
$8.1$&$3t\partial_t+(\rho\!+\!1)x\partial_x+(\rho\!-\!2)\partial_u$&$ \varphi(\omega)+\tfrac{\rho-2}3\ln t$&$xt^{-\frac{\rho+1}3}$&
$\delta\varphi'''''+\lambda\varphi'''+e^{\varphi}\varphi'-\tfrac{\rho+1}3\omega\varphi'+\tfrac{\rho-2}3=0$\\
\hline
$8.2$&$t\partial_t+a\partial_x-\partial_u$&$ \varphi(\omega)-\ln t$&$x\!-\!a\ln t$&
$\delta\varphi'''''+\lambda\varphi'''+e^{\varphi}\varphi'-a\varphi'-1=0$\\
\hline
$9$&$3\partial_t+x\partial_x+\partial_u$&$ \varphi(\omega)+\tfrac t3$&$xe^{-\frac t3}$&
$\delta\varphi'''''+\lambda\varphi'''+e^{\varphi}\varphi'-\tfrac13\omega\varphi'+\tfrac13=0$\\
\hline
\end{tabular}
\\[1.5ex]
\parbox{160mm}{In each row the reduced equation is denoted by $RL$ with the subscript given in the first column, e.g.\ $RL_{6.1}$ for row~$6.1$ and $RL_{5,a}$ for row~$5_a$. }
\end{center}
\end{table}

Several of the reduced equations can be integrated completely. The equation $RL_3$ is a first-order linear ODE for $\ln\varphi$ and integrates to
$\varphi=\exp\left(\dfrac{C-\int\frac{\beta(\omega)}{(\omega+a)^2} {\rm d}\omega-\int\frac{\gamma(\omega)}{(\omega+a)^4} {\rm d}\omega}{\omega+a}\right)$,
which gives the exact solution of the equation
$$L_3\colon\quad u_t+\ln(u)\,u_x+\beta(t)\,u_{xxx}+\gamma(t)\,u_{xxxxx}=0$$
in the form
$$u=\exp\left(\dfrac{x+C-\int\frac{\beta(t)}{(t+a)^2}{\rm d}t-\int\frac{\gamma(t)}{(t+a)^4}{\rm d}t}{t+a}\right).$$
The equation $RL_{4,\varepsilon}$ is the particular case of $RL_3$ with $\beta=\lambda\omega^2$, $\gamma=\delta\omega^4$ and $a=\varepsilon$, and its integration yields the exact solution of the equation
$$L_4\colon\quad u_t+\ln(u)\,u_x+\lambda t^2 u_{xxx}+\delta t^4 u_{xxxxx}=0$$
in the form
$$u=\exp\left({\frac{x+C+\lambda\left(\frac{\varepsilon^2}{t+\varepsilon}+2\varepsilon\ln(t+\varepsilon)-t\right)+\delta\left(\frac{\varepsilon^2}3\frac{18t^2+30\varepsilon t+13\varepsilon^2}{(t+\varepsilon)^3}+4\varepsilon\ln(t+\varepsilon)-t\right)}{t+\varepsilon}}\right).$$
Similarly, the solution $\varphi=\exp\left(\frac{C\omega^3+3\lambda\omega^2+\delta}{3\omega^4}\right)$ of the equation $RL_{5,0}$ leads to the exact solution of the equation
$$L_5\colon\quad u_t+\ln(u)\,u_x+\lambda u_{xxx}+\delta u_{xxxxx}=0$$
in the form
$$u=\exp\left(\frac{(3x+C)t^3+3\lambda t^2+\delta}{3t^4}\right).$$

The reductions presented in Cases 2, 6.1, 6.2 and 7 are valid also for the case $f(u)=u$ (or $n=1$).
Following~\cite{KPV2014}, where this case was investigated in detail, we list three additional inequivalent reductions of the equations~\eqref{eq:genKawahara_u}
for the completeness of the presented results.

The ansatz $u=\varphi(\omega)+\frac{x}{t+a},$ where $\omega=t,$ constructed using the subalgebra
$ \langle (t+a)\partial_x+\partial_u\rangle$ reduces the equation~\eqref{eq:genKawahara_u}
to the first order ODE  $(\omega+a)\varphi'+\varphi=0$.
Its solution $\varphi={C}/{(\omega+a)}$ leads to the so-called `degenerate' solution $u= {(x+C)}/{(t+a)}$ of Eq.~\eqref{eq:genKawahara_u}.

The ansatz $u=2t/a+\varphi(\omega)$, where  $\omega=x-{t^2}/a$, $a\neq0$, arising from the subalgebra $\langle a\partial_t+2t\partial_x+2\partial_u\rangle$, reduces the constant coefficient equation~\eqref{eq:genKawahara_u} with $\beta=\lambda$ and $\gamma=\delta$
 to the ODE
 $\delta\varphi'''''+\lambda\varphi'''+\varphi\varphi'+2/a=0$.

The last case concerns the equation
\begin{eqnarray}\label{eq12}
 u_t+ uu_x+ \lambda(t^2\!+\!1)^{\frac12} e^{3\nu\arctan t} u_{xxx}+\delta(t^2\!+\!1)^{\frac32} e^{ 5\nu\arctan t}  u_{xxxxx}=0,
\end{eqnarray}
 where $\nu$ is an arbitrary constant. The ansatz $u=\frac{e^{\nu\arctan t}}{\sqrt{t^2+1}}\varphi(\omega)+\frac{xt}{t^2+1}$, where $\omega=\frac{xe^{-\nu\arctan t}}{\sqrt{t^2+1}}$, constructed using the subalgebra $\langle(t^2+1)\partial_t+(t+\nu)x\partial_x+(x+(\nu-t)u)\partial_u\rangle$ reduces Eq.~\eqref{eq12}
to the ODE  $\delta\varphi'''''+\lambda \varphi'''+(\varphi-\nu\omega)\varphi'+\nu\varphi+\omega=0$.

\bigskip
The classification of inequivalent  Lie reductions of equations~\eqref{eq:genKawahara} (resp.~\eqref{eq_ggKawahara}) is completed.
The obtained reductions may be used also for solving the related invariant boundary value problems  (see~\cite{vane2015a,vane2014a} for the details).

\section{Conservation laws}\label{sec:CL}

We describe the low-order local conservation laws of the equations of the
class~\eqref{eq_ggKawahara}, namely, all local conservation laws whose density
order does not exceed two. Point transformations that realize the
gaugings of Section~\ref{sec:gauging} induce one-to-one mappings between
local conservation laws (and their characteristics) of the corresponding
equations. It therefore suffices to work within the two gauged normalized
subclasses, the class~\eqref{eq:genKawahara} with $f_{uu}\neq0$ and the
class~\eqref{eq:genKawahara_u} for the case $f_{uu}=0$; the result for an arbitrary
equation of the class~\eqref{eq_ggKawahara} is obtained by pulling the
conserved vector back with the corresponding equivalence transformation.

We look for conservation laws in characteristic form, following the direct construction method of~\cite{AncoBluman2002},
\[
  \Dt T+\Dx X=\Lambda\,
  \bigl(u_t+f(u)u_x+\beta u_{xxx}+\gamma u_{xxxxx}\bigr),
\]
with characteristics $\Lambda$ of order at most four and densities of
order at most two. Below the characteristics are denoted by $Q$. In other
words, we classify conservation laws whose density order in the sense
of~\cite{PopovychSergyeyev2010} does not exceed two. For the reduced
representative of a nontrivial conservation law, the characteristic
order does not exceed twice the density order and hence is at most four.
In contrast to evolution equations of even order $2q$, for which the
order of any local conservation law is bounded by~$q$
(see~\cite{PopovychSergyeyev2010} and references therein), no universal
bound exists for odd-order evolution equations. The existence of an
infinite sequence of conservation laws of increasing order is therefore
a strong indication of integrability. For the constant-coefficient case,
the classification of higher-order and generalized conservation laws is
given in~\cite{Vasicek}; no independent local conservation laws occur
there beyond the low-order ones. Throughout,
$F(u)=\int^u f(s)\,ds$, $G(u)=\int^u s f(s)\,ds$, and
$H(u)=\int^u F(s)\,ds$.

We now give a sketch of the determining-equation calculation. Modulo a
locally trivial density, a density of order at most two can be written as
\begin{equation}\label{eq:CLansatz}
  T=A(t)u_{xx}^2+B(t)u_x^2+W(t,x,u).
\end{equation}
Substituting
$u_t=-f u_x-\beta u_{xxx}-\gamma u_{xxxxx}$ into $\Dt T$ and imposing
$\Eu(\Dt T|_{\mathrm{sol}})=0$ yields, from the coefficients of the
highest independent jet monomials,
\begin{align}
  B&=-\frac{\beta}{\gamma}A,\label{eq:CLdet1}\\
  W_{uu}&=2\frac{A}{\gamma}f(u)+C_0(t)+C_1(t)x,
  \qquad C_1=-\frac{2A_t}{5\gamma},\label{eq:CLdet2}\\
  2\beta\gamma A_t+5A\bigl(\gamma\beta_t-\beta\gamma_t\bigr)&=0.
  \label{eq:CLdet3}
\end{align}
Integrating~\eqref{eq:CLdet2} twice with respect to $u$, and omitting a
term independent of $u$, gives
\begin{equation}\label{eq:CLWform}
  W=2\frac{A}{\gamma}H(u)
    +\frac12\bigl(C_0+C_1x\bigr)u^2+E(t,x)u,
  \qquad H''=f.
\end{equation}
The remaining determining equation is
\begin{equation}\label{eq:CLclassifying}
\begin{aligned}
&\beta E_{xxx}+\gamma E_{xxxxx}+E_t
  +f(u)\bigl(C_1u+E_x\bigr)
  +2\left(\frac{A}{\gamma}\right)_t F(u)+\bigl(C_0'+C_1'x\bigr)u =0.
\end{aligned}
\end{equation}
Equations~\eqref{eq:CLdet1}--\eqref{eq:CLclassifying} provide a compact
complete determining system for the densities under consideration. Its
splitting gives the following alternatives.
\begin{enumerate}\itemsep2pt
\item If $A=0$ and the arbitrary elements are unrestricted, only constant
choices of $C_0$ and $E$ survive, yielding the two universal
laws~\eqref{eq:CLuniversal}.
\item If $A\neq0$ and $\beta,\gamma$ are constant, the system yields the
autonomous energy for arbitrary $f$.
\item If $f_{uu}\neq0$ and the dispersive coefficients are genuinely
time-dependent, compatibility of~\eqref{eq:CLdet3} and
\eqref{eq:CLclassifying} restricts $f$, up to $G^\sim$-equivalence, to a
power, an exponential, or a logarithm. For the power family the generic
scaling branch is supplemented by the exceptional reciprocal-power case
$n=-1$.
\item If $f_{uu}=0$, the system gives the Galilean law for arbitrary
$\beta,\gamma$ and, when $A\neq0$, the additional conditions
$\beta=\beta_0\gamma^{1/3}$ and
$\bigl(\gamma^{2/3}\bigr)'''=0$ defining the linear energy branch.
\end{enumerate}
Each resulting characteristic and density is stated below in closed
form. All displayed conserved vectors have been verified directly by
computing $\Eu(\Dt T|_{\mathrm{sol}})$ and $\Dt T+\Dx X$ symbolically. A
convenient tool for such computations is the package GeM for
Maple~\cite{Cheviakov}.

\paragraph{Universal conservation laws ($f,\alpha,\beta,\gamma$ arbitrary).}
Every equation of the class~\eqref{eq_ggKawahara} possesses the two
conservation laws
\begin{equation}\label{eq:CLuniversal}
  Q_1=1,\quad T_1=u,\qquad\qquad
  Q_2=u,\quad T_2=\tfrac12u^2,
\end{equation}
with the fluxes
\begin{align}
  X_1 &= \alpha F(u)+\beta\,u_{xx}+\gamma\,u_{xxxx}, \label{eq:CLX1}\\
  X_2 &= \alpha G(u)+\beta\bigl(u\,u_{xx}-\tfrac12u_x^2\bigr)
         +\gamma\bigl(u\,u_{xxxx}-u_x u_{xxx}+\tfrac12u_{xx}^2\bigr).
  \label{eq:CLX2}
\end{align}
Here $T_1$ is the mass (total wave amplitude) and $2\int T_2\,dx$ is the
$L^2$-norm squared; both fluxes are exact for the entire class, with the
nonlinearity~$f$ and the coefficient~$\alpha$ entering only through the terms
$\alpha F$ and $\alpha G$. The dispersive contributions to $X_1,X_2$ are
universal. In the gauge $\alpha=1$ used below, the factor $\alpha$ in these
fluxes disappears.

The remaining conservation laws with $\Lambda$ of order up to four appear
only for distinguished forms of $f,\beta,\gamma$, and these forms are
exactly the ones singled out by the Lie symmetry classification of
Tables~\ref{Kawahara_tab:1} and~\ref{Kawahara_tab:2}. Below they are
presented in the gauge $\alpha=1$ (with, in addition, $\kappa_0=0$ for the
subclass with $f_{uu}=0$), that is, for the classes~\eqref{eq:genKawahara}
and~\eqref{eq:genKawahara_u}; the corresponding laws for the original
class~\eqref{eq_ggKawahara} are recovered by the pullback described at the
beginning of this section. We list them by the form of $f$.

\paragraph{Autonomous energy (Case~2 of Table~\ref{Kawahara_tab:1};
$\beta=\lambda$, $\gamma=\delta$ constant, arbitrary $f$).}
Whenever the coefficients are constant, the equation is autonomous and
admits the Hamiltonian energy
\begin{equation}\label{eq:CLenergyconst}
  Q=\delta\,u_{xxxx}+\lambda\,u_{xx}+F(u),\qquad
  T=\tfrac12\bigl(\delta\,u_{xx}^2-\lambda\,u_x^2\bigr)+H(u),
  \quad H''=f,
\end{equation}
where $H(u)=\int^{u}\!F(s)\,ds$. This law exists for any
nonlinearity $f$. For $f=u^n$,
\[
 H=\frac{u^{n+2}}{(n+1)(n+2)}\quad(n\neq-1,-2),\qquad
 H=u\ln u -u\quad(n=-1),\qquad H=-\ln u\quad(n=-2).
\]
For $f=e^u$ one has $H=e^u$, whereas for $f=\ln u$,
$H=\tfrac12u^2\ln u-\tfrac34u^2$ (the last term is a multiple of the
quadratic density $T_2$).

\paragraph{Power nonlinearity $f=u^n$ ($n\neq0,1$).}
For constant $\beta=\lambda$ and $\gamma=\delta$, the autonomous
energy~\eqref{eq:CLenergyconst} is admitted for every $n\neq0,1$.
Apart from this autonomous case, the nonconstant-coefficient laws split
as follows.

If $n\neq-1$, the coefficients belong to the one-parameter subfamily of
Case~6.1,
\begin{equation}\label{eq:powerbranch}
  \beta=\lambda t^\rho,\qquad
  \gamma=\delta t^{(4-n)/(n+1)},\qquad
  \rho=\frac{2-n}{n+1}.
\end{equation}
The characteristic is
\begin{equation}\label{eq:powerQ}
  Q=(n+1)t\bigl(\gamma u_{xxxx}+\beta u_{xx}\bigr)
    +t u^{n+1}-xu.
\end{equation}
For $n\neq-2$ a corresponding density is
\begin{equation}\label{eq:powerT}
  T=\frac{n+1}{2}t\bigl(\gamma u_{xx}^2-\beta u_x^2\bigr)
    +\frac{t}{n+2}u^{n+2}-\frac12xu^2.
\end{equation}
The exponent $n=-2$ belongs to the same coefficient branch, namely
\begin{equation}\label{eq:powerMinusTwoBranch}
  \beta=\lambda t^{-4},\qquad \gamma=\delta t^{-6},
\end{equation}
but the potential term in~\eqref{eq:powerT} has to be replaced by its
logarithmic limit. In this case
\begin{equation}\label{eq:powerMinusTwoQT}
\begin{aligned}
  Q&=-t\bigl(\gamma u_{xxxx}+\beta u_{xx}\bigr)+\frac{t}{u}-xu,\\
  T&=-\frac12t\bigl(\gamma u_{xx}^2-\beta u_x^2\bigr)
     +t\ln u -\frac12xu^2.
\end{aligned}
\end{equation}

The exceptional reciprocal-power value $n=-1$ does not arise from
\eqref{eq:powerbranch}, since that formula divides by $n+1$. Instead, it
selects the exponential-coefficient Case~7,
\begin{equation}\label{eq:reciprocalbranch}
  f=u^{-1},\qquad \beta=\lambda e^t,\qquad
  \gamma=\delta e^{5t/3},
\end{equation}
for which an additional conservation law is
\begin{equation}\label{eq:reciprocalQT}
\begin{aligned}
  Q={}&3\gamma u_{xxxx}+3\beta u_{xx}+3\ln u -xu+t,\\
  T={}&\frac32\bigl(\gamma u_{xx}^2-\beta u_x^2\bigr)
       +3\bigl(u\ln u -u\bigr)-\frac12xu^2+tu.
\end{aligned}
\end{equation}
These formulas are understood locally on domains where $u\neq0$.
There are no other second-order conservation laws for the power
nonlinearity with nonconstant $\beta(t),\gamma(t)$.

\paragraph{Exponential nonlinearity $f=e^u$.}
A second-order conservation law exists for constant $\beta,\gamma$
(the autonomous energy~\eqref{eq:CLenergyconst} with $H=e^u$) and,
in addition, for the subfamily of Case~8.2,
\begin{equation}\label{eq:expbranch}
  \beta=\lambda\,t^{-1},\qquad \gamma=\delta\,t^{-1},
\end{equation}
with
\begin{equation}\label{eq:expQT}
  Q=2\bigl(\delta\,u_{xxxx}+\lambda\,u_{xx}+t\,e^{u}-x\bigr),\qquad
  T=\delta\,u_{xx}^2-\lambda\,u_x^2+2t\,e^{u}-2x\,u.
\end{equation}
Under the Type~A reparameterization $n=\sigma\tilde n$ of
Section~\ref{sec:contractions} below, applied together with the associated
transformation of the dependent variable, the power coefficient
branch~\eqref{eq:powerbranch} tends to the branch~\eqref{eq:expbranch} as
$n\to\infty$, since $\rho=(2-n)/(n+1)\to-1$; this is the conservation-law
counterpart of the Type~A contraction.

\paragraph{Logarithmic nonlinearity $f=\ln u$.}
A second-order conservation law exists for constant $\beta,\gamma$
(the autonomous energy~\eqref{eq:CLenergyconst}) and, in addition,
for the coefficients of Case~4 of Table~\ref{Kawahara_tab:1},
\begin{equation}\label{eq:logbranch}
  \beta=\lambda\,t^{2},\qquad \gamma=\delta\,t^{4},
\end{equation}
with
\begin{equation}\label{eq:logQT}
\begin{aligned}
  Q&=2\bigl(\delta t^5u_{xxxx}+\lambda t^3u_{xx}
       +t u\ln u-xu\bigr),\\
  T&=\delta t^5u_{xx}^2-\lambda t^3u_x^2
    +t u^2\ln u-\bigl(x+\tfrac{t}{2}\bigr)u^2.
\end{aligned}
\end{equation}

\paragraph{Linear nonlinearity $f=u$ (class~\eqref{eq:genKawahara_u}).}
Two families arise. First, for arbitrary $\beta(t),\gamma(t)$ there is
the Galilean conservation law
\begin{equation}\label{eq:linGal}
  Q=x-t\,u,\qquad T=x\,u-\tfrac12\,t\,u^2,
\end{equation}
whose flux is
\begin{equation}\label{eq:linGalX}
\begin{aligned}
  X = \tfrac12\,x\,u^2 - \tfrac13\,t\,u^3
    &+ \beta\bigl(x\,u_{xx}-u_x-t\,u\,u_{xx}+\tfrac12 t\,u_x^2\bigr)\\
    &+ \gamma\bigl(x\,u_{xxxx}-u_{xxx}-t\,u\,u_{xxxx}+t\,u_x u_{xxx}
       -\tfrac12 t\,u_{xx}^2\bigr).
\end{aligned}
\end{equation}
This law corresponds to the kernel generator $t\partial_x+\partial_u$ of
Table~\ref{Kawahara_tab:2}. Second, a dispersive energy law exists
if and only if the coefficients satisfy
\begin{equation}\label{eq:linenergycond}
  \beta=\beta_0\,\gamma^{1/3},\qquad
  \bigl(\gamma^{2/3}\bigr)'''=0,
\end{equation}
that is, $\gamma^{2/3}$ is a quadratic polynomial in $t$. Since $\gamma$ is
only assumed to be nonvanishing and may be negative, fractional powers of
$\gamma$ are interpreted here via the real cube root, so that
$P(t):=\gamma^{2/3}=|\gamma|^{2/3}>0$ and $\gamma^{5/3}=\gamma P$;
equivalently, $\gamma=\hat\delta\,(a\,t^2+b\,t+c)^{3/2}$ with a constant sign
factor $\hat\delta\in\{-1,1\}$ and the quadratic $a\,t^2+b\,t+c=P$ positive on
the interval considered. With this convention the characteristic and density
can be chosen as
\begin{equation}\label{eq:linenergyT}
\begin{aligned}
  Q={}&2P\bigl(\gamma u_{xxxx}+\beta u_{xx}\bigr)
       +Pu^2-P'xu+\tfrac12P''x^2,\\
  T={}&\gamma^{5/3}u_{xx}^2-\beta_0\gamma u_x^2
    +\tfrac13Pu^3-\tfrac12P'xu^2+\tfrac12P''x^2u.
\end{aligned}
\end{equation}
The condition~\eqref{eq:linenergycond} is covariant under the equivalence
group $G^\sim_1$ of the class~\eqref{eq:genKawahara_u}, whose action on
$t$ is by M\"obius transformations; the three $G^\sim_1$-inequivalent
canonical representatives, distinguished by the sign of the discriminant
of the quadratic $P$, are, up to the constant sign factor $\hat\delta$,
\begin{equation}\label{eq:linenergyreps}
  \gamma=\mathrm{const}\ \ (\text{Case }3),\qquad
  \gamma=t^{3/2}\ \ (\text{Case }1|_{\rho=1/2}),\qquad
  \gamma=(t^2+1)^{3/2}\ \ (\text{Case }4|_{\nu=0}),
\end{equation}
which coincide with the corresponding cases of
Table~\ref{Kawahara_tab:2}. In particular, the irreducible
representative $\gamma=(t^2+1)^{3/2}$ carries a genuine conservation law;
it is absent from a list restricted to the reducible power-type
solutions $\gamma=(c_1t+c_0)^{3/2}$, $(c_1t+c_0)^{3}$, $\mathrm{const}$.

\paragraph{Fluxes.}
The fluxes~\eqref{eq:CLX1}, \eqref{eq:CLX2} and~\eqref{eq:linGalX} are
recorded in closed form above. For the energy-type
densities~\eqref{eq:CLenergyconst}, \eqref{eq:powerT},
\eqref{eq:powerMinusTwoQT}, \eqref{eq:reciprocalQT}, \eqref{eq:expQT},
\eqref{eq:logQT} and~\eqref{eq:linenergyT}, the fluxes are obtained by
inverting the total derivative $\Dx$ on $-\Dt T|_{\text{sol}}$. All of
them are conveniently expressed through the universal energy flux block
\begin{equation}\label{eq:CLXE}
\begin{aligned}
  \mathcal X_E={}&\tfrac12F^2
  +\beta\bigl(F u_{xx}-f u_x^2\bigr)
  +\gamma\bigl(F u_{xxxx}+f u_{xx}^2-f u_x u_{xxx}+f_u u_x^2 u_{xx}\bigr)\\
  &+\tfrac12\beta^2\bigl(u_{xx}^2-2u_x u_{xxx}\bigr)
  +\beta\gamma\bigl(2u_{xx}u_{xxxx}-u_{xxx}^2-u_x u_{xxxxx}\bigr)\\
  &+\tfrac12\gamma^2\bigl(2u_{xx}u_{xxxxxx}-2u_{xxx}u_{xxxxx}+u_{xxxx}^2\bigr),
\end{aligned}
\end{equation}
where $F_u=f$, and the lower-order correction
\begin{equation}\label{eq:CLY}
  \mathcal Y=\beta\,u u_x+\gamma\bigl(u u_{xxx}-2u_x u_{xx}\bigr);
\end{equation}
in~\eqref{eq:CLXE} and~\eqref{eq:CLY} the functions $\beta$ and $\gamma$
stand for the coefficients of the corresponding equation, and $X_1$,
$X_2$ below denote the fluxes~\eqref{eq:CLX1}, \eqref{eq:CLX2} with the
same $\beta$, $\gamma$ and with $F$, $G$ specified by the corresponding
nonlinearity. The fluxes associated with the displayed energy-type densities are then
\begin{gather}
  X=\mathcal X_E
  \qquad\text{for~\eqref{eq:CLenergyconst}, with arbitrary $f$ and
  $\beta=\lambda$, $\gamma=\delta$;}\label{eq:XenergyConst}\\
  X=(n+1)t\,\mathcal X_E-xX_2+\mathcal Y
  \quad\text{for~\eqref{eq:powerT}, }n\neq-1,-2,0,1,\quad
  F=\frac{u^{n+1}}{n+1},\quad G=\frac{u^{n+2}}{n+2};
  \label{eq:XenergyPower}\\
  X=-t\,\mathcal X_E-xX_2+\mathcal Y
  \quad\text{for~\eqref{eq:powerMinusTwoQT}, with }
  f=u^{-2},\ F=-u^{-1},\ G=\ln u ;
  \label{eq:XenergyPowerMinusTwo}\\
  X=3\mathcal X_E-xX_2+tX_1+\mathcal Y
  \quad\text{for~\eqref{eq:reciprocalQT}, with }
  f=u^{-1},\ F=\ln u ,\ G=u;
  \label{eq:XenergyReciprocal}\\
  X=2t\,\mathcal X_E-2x\,X_1+2\bigl(\beta u_x+\gamma u_{xxx}\bigr)
  \qquad\text{for~\eqref{eq:expQT}, with }f=F=e^u;
  \label{eq:XenergyExp}\\
  X=2t\,\mathcal X_E+2(t-x)X_2+2\mathcal Y
  \qquad\text{for~\eqref{eq:logQT}, with }
  f=\ln u,\ F=u\ln u-u,\ G=\tfrac12u^2\ln u-\tfrac14u^2;
  \label{eq:XenergyLog}\\
\begin{aligned}
  X={}&2P\,\mathcal X_E-P'x\,X_2+\tfrac12P''x^2X_1
  +P'\,\mathcal Y+P''\bigl(\beta u+\gamma u_{xx}\bigr)
  -P''x\bigl(\beta u_x+\gamma u_{xxx}\bigr)\\
  &\text{for~\eqref{eq:linenergyT}, with }
  f=u,\ F=\tfrac12u^2,\ G=\tfrac13u^3,\
  \gamma=P^{3/2},\ \beta=\beta_0\gamma^{1/3}.
\end{aligned}\label{eq:XenergyLin}
\end{gather}
All the expressions~\eqref{eq:XenergyConst}--\eqref{eq:XenergyLin} have
been verified symbolically. In the formal limit $\delta\to0$ of the
constant-coefficient case~\eqref{eq:XenergyConst}, which leads outside the
class under consideration, the flux reduces to the classical energy flux of
the Korteweg--de Vries-type equations,
$X=\tfrac12F^2+\lambda(Fu_{xx}-fu_x^2)+\tfrac12\lambda^2(u_{xx}^2-2u_xu_{xxx})$.

\paragraph{Global balance and physical meaning.}
Integrating $\Dt T+\Dx X=0$ over a domain $\Omega\subseteq\mathbb R$ on
which the flux vanishes at the boundary (e.g., for $u$ decaying at
infinity) gives $\tfrac{d}{dt}\int_\Omega T\,dx=0$. Thus $T_1$ gives the
conserved mass and $T_2$ gives one half of the squared $L^2$ norm. The
autonomous and nonautonomous energy-type densities displayed above
balance the third- and fifth-order dispersive contributions
$-\beta u_x^2$ and $\gamma u_{xx}^2$ against a nonlinear potential and,
when present, moment terms. The Galilean law~\eqref{eq:linGal} separately
tracks the linear-in-$x$ moment of the profile.

Although the equations under consideration are not Euler--Lagrange
equations and their point symmetries do not automatically generate
conservation laws, the distinguished logarithmic and exponential
symmetry-extension cases also carry genuine second-order conservation
laws, and these organize themselves across the nonlinearities by the same
contractions as the symmetry algebras
(Section~\ref{sec:contractions}). The results agree with, and extend,
the conservation laws reported for particular subclasses
in~\cite{KPV2014, GandariasRosaRecioAnco}.

\begin{remark}
Low-order local conservation laws of the class~\eqref{eq_ggKawahara}
were also considered in~\cite{GandariasRosaRecioAnco} by the direct
multiplier method, using the gauge in which the coefficient of
$u_{xxxxx}$ is one instead of the gauge $\alpha=1$ used here. The two
descriptions are related by the time reparameterization
$\tilde t=-\int\gamma(t)\,dt$. After this gauge conversion, the universal
laws~\eqref{eq:CLuniversal}, the autonomous energy
\eqref{eq:CLenergyconst}, the generic power-law branch
\eqref{eq:powerbranch}--\eqref{eq:powerT} for $n\neq-1,-2$, its
$n=-2$ limit~\eqref{eq:powerMinusTwoBranch}--\eqref{eq:powerMinusTwoQT}, the Galilean law
\eqref{eq:linGal}, and the linear energy branch
\eqref{eq:linenergycond}--\eqref{eq:linenergyT} agree with the
corresponding results of~\cite{GandariasRosaRecioAnco}. Their integral
form of the power-law density also covers $n=-2$, since the nonlinear
potential is then interpreted as $\ln u$.

The exceptional reciprocal exponent $n=-1$ is not covered by the
power-law formula in~\cite{GandariasRosaRecioAnco}, which contains a
factor $(n+1)^{-1}$, and no separate law was given there. The
nonautonomous logarithmic and exponential laws are absent there as well.
Consequently, the reciprocal-power law~\eqref{eq:reciprocalQT}, the
logarithmic law~\eqref{eq:logQT}, and the exponential law
\eqref{eq:expQT} are the three additional inequivalent nonlinear
branches supplied by the present classification. The logarithmic and
exponential coefficient branches are limiting positions of the generic
power family: the exponents
$\rho=(2-n)/(n+1)$ and $(4-n)/(n+1)$ tend respectively to $(2,4)$ as
$n\to0$ and to $(-1,-1)$ as $n\to\infty$.

Conservation laws for several constant-coefficient Kawahara-type
equations were constructed earlier in~\cite{Leites} by the nonlinear
self-adjointness technique and Ibragimov's theorem. The conserved vectors
obtained there are symmetry-generated counterparts of the mass,
quadratic and energy laws above and reduce to them up to locally trivial
(total-derivative) terms, consistently with that construction yielding
particular conservation laws rather than a complete classification.
\end{remark}

\section{Contractions}
\label{sec:contractions}

Nontrivial limit processes between equations admitting Lie symmetry
extensions of different structure have been known for a long time. In
particular, in the group classification of nonlinear diffusion and
reaction--diffusion equations the cases with exponential nonlinearities
can be derived from those with power nonlinearities by means of such
limit processes. By analogy with the theory of Lie algebras, where the
corresponding limits between algebras go back to In\"on\"u and Wigner,
such connections between classification cases are called contractions.
A theoretical background on contractions of differential equations, of
their Lie invariance algebras and of their solutions was first discussed
in~\cite{IPS2007II}; contractions accompanying group classification
within the equivalence-based framework were studied for variable
coefficient mKdV and reaction--diffusion equations
in~\cite{vane2012a,VPS2012}. We say that a case of Lie symmetry
extension contracts to another case if there exists a one-parameter
family of equations from the class, obtained by an appropriate
reparameterization of the variables and of the constant parameters
involved, such that in the corresponding limit the equations, their
maximal Lie invariance algebras, the associated ans\"atze and the
reduced equations transform consistently. For the contractions studied
below, in which the full maximal Lie invariance algebra of the source is
required to contract onto the full maximal algebra of the target, the
two dimensions must coincide.  This requirement restricts
the candidate pairs to the two types indicated in
Figure~\ref{fig:classification_final}. Throughout this section we work
exclusively with the gauged classes~\eqref{eq:genKawahara} and
\eqref{eq:genKawahara_u}. Thus, the coefficient of the nonlinear
transport term is fixed to one in both classes (and $\kappa_0=0$ in
the second class), and all contraction arrows refer directly to the
canonical cases of Tables~\ref{Kawahara_tab:1} and
\ref{Kawahara_tab:2}. In what follows we
use the notation $M.k\to N.l$ for the contraction of Case~$k$ of
Table~$M$ to Case~$l$ of Table~$N$; $\lambda$ and $\delta$ denote
fixed nonzero constants from the classification tables, whereas
$\sigma$ and $\varepsilon$ denote the contraction parameters.

\subsection*{Type~A: \texorpdfstring{$u^n\to e^u$}{u\^{}n to e\^{}u}}

Contractions of this type relate the cases with power nonlinearities to
those with exponential ones, namely, Cases~6--7 to Cases~8--9 of
Table~\ref{Kawahara_tab:1}, and they are induced by a family of
equivalence transformations. Consider the family of point
transformations parameterized by a constant $\sigma>0$,
\begin{equation}\label{eq:typeA}
  \tilde t=t,\qquad \tilde x=x,\qquad \tilde u=\sigma(u-1),
\end{equation}
which belong to the equivalence group of the
class~\eqref{eq:genKawahara}; the arbitrary elements and the constants
involved are transformed, wherever this is relevant, according to
\[
  \tilde f(\tilde u)=f(u),\qquad \tilde\beta=\beta,\qquad
  \tilde\gamma=\gamma,\qquad \tilde n=\frac n\sigma,\qquad
  \tilde\rho=\rho.
\]
We apply transformation~\eqref{eq:typeA} to an equation of Case~1.6
(resp.~1.7) and proceed to the limit $\sigma\to+\infty$. The
nonlinearity then behaves as
\[
  f(u)=u^{n}=\Bigl(1+\frac{\tilde u}{\sigma}\Bigr)^{\sigma\tilde n}
  \ \xrightarrow{\ \sigma\to+\infty\ }\ e^{\tilde n\tilde u},
\]
while the coefficients $\beta$ and $\gamma$ are not affected; the
residual constant $\tilde n$ can additionally be set equal to~1 by a
transformation from~$G^\sim$. Therefore, the limit equation is
exactly the equation of Case~1.8 (resp.~1.9) with the same value
of~$\rho$. The same limit process establishes contractions of the
corresponding maximal Lie invariance algebras as algebras of vector
fields. Indeed, $\partial_u=\sigma\,\partial_{\tilde u}$, so that
$u\partial_u=(\sigma+\tilde u)\partial_{\tilde u}$. After applying the
equivalence transformation that sets $\tilde n=1$, that is, with
$n=\sigma$, dividing the basis operators of Cases~1.6 and~1.7 by $n$ we
obtain in the limit
\begin{align*}
  3nt\partial_t+(\rho+1)nx\partial_x+(\rho-2)u\partial_u
  &\ \longrightarrow\ 3t\partial_t+(\rho+1)x\partial_x+(\rho-2)\partial_{\tilde u},\\
  3n\partial_t+nx\partial_x+u\partial_u
  &\ \longrightarrow\ 3\partial_t+x\partial_x+\partial_{\tilde u},
\end{align*}
i.e., exactly the basis operators of Cases~1.8 and~1.9, respectively.
This results in the contractions
\[
  1.6\to1.8,\qquad 1.7\to1.9.
\]

Contractions of equations and of their Lie invariance algebras induce
contractions of the related objects, in particular, of ans\"atze,
reduced equations and solutions. Consider the contraction $1.6\to1.8$ in
detail. Setting $\tilde\varphi=\sigma(\varphi-1)$, we map the ansatz of
Case~6.1, $u=t^{(\rho-2)/(3n)}\varphi(\omega)$ with
$\omega=xt^{-(\rho+1)/3}$, by transformation~\eqref{eq:typeA} to
\[
  \tilde u=\sigma\bigl(t^{(\rho-2)/(3n)}\varphi-1\bigr)
  =\tilde\varphi(\omega)+\frac{\rho-2}{3\tilde n}\ln t+O(\sigma^{-1})
  \ \xrightarrow{\ \sigma\to+\infty\ }\
  \tilde\varphi(\omega)+\frac{\rho-2}{3\tilde n}\ln t,
\]
which for $\tilde n=1$ is exactly the ansatz of Case~8.1 with the same
invariant variable~$\omega$. The reduced equation $RL_{6.1}$ is mapped
by the same substitution, $\varphi=1+\tilde\varphi/\sigma$ and
$n=\sigma\tilde n$, after the multiplication by $\sigma$, to the
equation
\[
  \delta\tilde\varphi'''''+\lambda\tilde\varphi'''
  +\Bigl(1+\frac{\tilde\varphi}{\sigma}\Bigr)^{\sigma\tilde n}\tilde\varphi'
  -\frac{\rho+1}{3}\,\omega\tilde\varphi'
  +\frac{\rho-2}{3\tilde n}\Bigl(1+\frac{\tilde\varphi}{\sigma}\Bigr)=0,
\]
and the limit $\sigma\to+\infty$ results in the equation
\[
  \delta\tilde\varphi'''''+\lambda\tilde\varphi'''
  +e^{\tilde n\tilde\varphi}\tilde\varphi'
  -\frac{\rho+1}{3}\,\omega\tilde\varphi'+\frac{\rho-2}{3\tilde n}=0,
\]
which for $\tilde n=1$ coincides with $RL_{8.1}$. Analogously, the
ansatz $u=e^{t/(3n)}\varphi(\omega)$ with $\omega=xe^{-t/3}$ of Case~7
contracts to the ansatz $\tilde u=\tilde\varphi(\omega)+t/3$ of Case~9,
and the reduced equation $RL_{7}$ contracts to $RL_{9}$.

Contractions of conservation laws can be considered at the level of
characteristics, of densities and fluxes, or of the corresponding
conserved quantities, allowing normalizations that depend on the
contraction parameter. The contractions constructed above agree with the
classification of conservation laws in Section~\ref{sec:CL}. In contrast
to the symmetry contractions, where $\rho$ is fixed, the
conservation-law contraction acts along the family~\eqref{eq:typeA}
restricted to the branch~\eqref{eq:powerbranch}, on which the exponent
$\rho=(2-n)/(n+1)$ is tied to~$n$. Setting $\tilde n=1$, i.e.\ $n=\sigma$,
both coefficients vary along the family,
$\beta_\sigma=\lambda t^{(2-\sigma)/(\sigma+1)}$ and
$\gamma_\sigma=\delta t^{(4-\sigma)/(\sigma+1)}$, and
$(\beta_\sigma,\gamma_\sigma)\to(\lambda t^{-1},\delta t^{-1})$ as
$\sigma\to+\infty$, which is exactly the branch~\eqref{eq:expbranch}.

The characteristics transform as follows. In terms of
$\tilde u=\sigma(u-1)$, the left-hand side of the source equation
factorizes as
\[
  u_t+u^nu_x+\beta_\sigma u_{xxx}+\gamma_\sigma u_{xxxxx}
  =\frac1\sigma\Bigl(\tilde u_t
   +\Bigl(1+\frac{\tilde u}{\sigma}\Bigr)^{\sigma}\tilde u_x
   +\beta_\sigma\tilde u_{xxx}+\gamma_\sigma\tilde u_{xxxxx}\Bigr),
\]
so a conserved current $(T,X)$ with the characteristic $Q$ acquires the
characteristic $Q/\sigma$ with respect to the transformed equation, and
after the normalization of the conservation law by the factor
$2\sigma/\tilde n=2\sigma$ the transformed characteristic equals
$2Q$ with $Q$ given by~\eqref{eq:powerQ}. Since the derivatives transform
as $u_{x}=\tilde u_{x}/\sigma$, etc., we obtain, using
$t\gamma_\sigma=\delta t^{5/(\sigma+1)}\to\delta$ and
$t\beta_\sigma=\lambda t^{3/(\sigma+1)}\to\lambda$,
\[
  2Q=2\,\frac{\sigma+1}{\sigma}\,t
     \bigl(\gamma_\sigma\tilde u_{xxxx}+\beta_\sigma\tilde u_{xx}\bigr)
  +2t\Bigl(1+\frac{\tilde u}{\sigma}\Bigr)^{\sigma+1}
  -2x-\frac{2x\tilde u}{\sigma}
  \ \xrightarrow{\ \sigma\to+\infty\ }\
  2\bigl(\delta\tilde u_{xxxx}+\lambda\tilde u_{xx}+te^{\tilde u}-x\bigr),
\]
which is exactly the characteristic of the exponential
law~\eqref{eq:expQT}. Similarly, the normalized density~\eqref{eq:powerT}
becomes
\[
  2\sigma T=\frac{\sigma+1}{\sigma}\,t
   \bigl(\gamma_\sigma\tilde u_{xx}^2-\beta_\sigma\tilde u_x^2\bigr)
  +\frac{2\sigma t}{\sigma+2}\Bigl(1+\frac{\tilde u}{\sigma}\Bigr)^{\sigma+2}
  -\sigma x-2x\tilde u-\frac{x\tilde u^2}{\sigma}.
\]
The summand $-\sigma x$ constitutes the trivial conservation law with the
conserved current $(-\sigma x,0)$ and zero characteristic, and modulo
this summand
\[
  2\sigma T\ \xrightarrow{\ \sigma\to+\infty\ }\
  \delta\tilde u_{xx}^2-\lambda\tilde u_x^2+2te^{\tilde u}-2x\tilde u.
\]
Thus the density~\eqref{eq:powerT}, normalized by the factor
$2\sigma/\tilde n$, contracts, up to a summand corresponding to a trivial
conservation law, to the exponential density~\eqref{eq:expQT}, and the
associated characteristics contract accordingly, with no trivial
correction required.
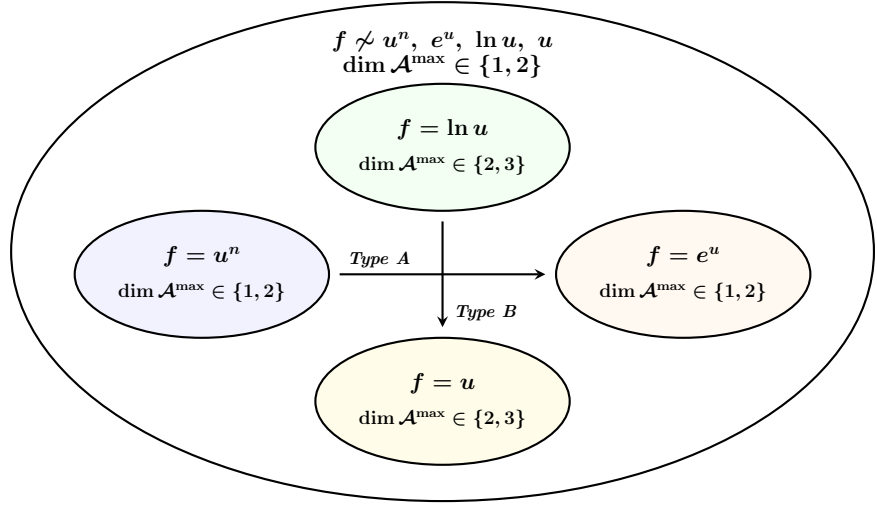
\begin{figure}[t!]
  \centering
  \begin{tikzpicture}[scale=0.6, transform shape]

    \tikzset{
      smooth curve/.style={thick, black},
      math label/.style={
        font=\Large\bfseries,
        execute at begin node=\boldmath
      },
      sub label/.style={
        font=\large\bfseries,
        execute at begin node=\boldmath
      },
      contraction arrow/.style={
        ->,
        >=stealth,
        thick,
        shorten >=4pt,
        shorten <=4pt,
        black
      }
    }

    \draw[smooth curve] (0, 0) ellipse (9.5 and 5.5);

    \node[math label, align=center] at (0, 4.35)
      {$f\not\sim u^n,\ e^u,\ \ln u,\ u$\\[-1mm]
       $\dim \mathcal{A}^{\max}\in\{1,2\}$};


    \draw[smooth curve, fill=green!5] (0, 2.3) ellipse (2.8 and 1.4);
    \node[math label] at (0, 2.7) {$f = \ln u$};
    \node[sub label] at (0, 1.9) {$\dim \mathcal{A}^{\max} \in\{2, 3\}$};

    \draw[smooth curve, fill=blue!5] (-5.3, -0.5) ellipse (2.8 and 1.4);
    \node[math label] at (-5.3, -0.1) {$f = u^n$};
    \node[sub label] at (-5.3, -0.9) {$\dim \mathcal{A}^{\max} \in\{1, 2\}$};

    \draw[smooth curve, fill=orange!5] (5.3, -0.5) ellipse (2.8 and 1.4);
    \node[math label] at (5.3, -0.1) {$f = e^u$};
    \node[sub label] at (5.3, -0.9) {$\dim \mathcal{A}^{\max} \in\{1, 2\}$};

    \draw[smooth curve, fill=yellow!10] (0, -3.3) ellipse (2.8 and 1.4);
    \node[math label] at (0, -2.9) {$f = u$};
    \node[sub label] at (0, -3.7) {$\dim \mathcal{A}^{\max} \in\{2, 3\}$};

    \draw[contraction arrow] (-2.5, -0.5) -- (2.5, -0.5)
        node[pos=0.22, anchor=south, font=\itshape\bfseries, text=black]
        {Type A};

    \draw[contraction arrow] (0, 0.9) -- (0, -1.9)
        node[pos=0.8, anchor=west, font=\itshape\bfseries, text=black,
        xshift=3pt] {Type B};

  \end{tikzpicture}
  \caption{Structure of the class with respect to the functional form of $f$,
  the possible dimensions of the maximal Lie invariance algebras, and the
  Type~A and Type~B contractions. The four inner ovals represent, up to
  equivalence transformations, the power, exponential, logarithmic, and
  linear subclasses. The complementary region consists of all other
  functional forms of $f$.}
  \label{fig:classification_final}
\end{figure}

\subsection*{Type~B: \texorpdfstring{$\ln u\to u$}{ln u to u}}

Contractions of this type relate Cases~3--5 of
Table~\ref{Kawahara_tab:1} to Cases~0, $1|_{\rho=2}$ and~3 of
Table~\ref{Kawahara_tab:2}. Both tables are written in the gauge in
which the coefficient of the nonlinear transport term equals one. We
therefore stay entirely within the gauged classes
\eqref{eq:genKawahara} and~\eqref{eq:genKawahara_u} and introduce no
additional coefficient in front of the nonlinear term.

Let $\varepsilon\ne0$. For fixed functions $\beta(t)$ and $\gamma(t)$,
consider the logarithmic equation
\begin{equation}\label{eq:typeBsource}
 u_t+\ln u\,u_x+\varepsilon^3\beta(t)u_{xxx}
 +\varepsilon^5\gamma(t)u_{xxxxx}=0.
\end{equation}
For generic $\beta$ and $\gamma$, equation~\eqref{eq:typeBsource}
belongs to Case~1.3. Taking
\[
 \beta(t)=\lambda t^2,\qquad \gamma(t)=\delta t^4
\]
gives a family in Case~1.4, whereas taking $\beta=\lambda$ and
$\gamma=\delta$ gives a family in Case~1.5. The factors
$\varepsilon^3$ and $\varepsilon^5$ merely rescale the nonzero
constants occurring in the corresponding rows of
Table~\ref{Kawahara_tab:1}.

Apply the one-parameter family of point transformations
\begin{equation}\label{eq:typeB}
 \tilde t=t,\qquad \tilde x=\frac{x}{\varepsilon},\qquad
 \tilde u=\frac{u-1}{\varepsilon}.
\end{equation}
For each fixed $\varepsilon\ne0$, transformation~\eqref{eq:typeB} is an
equivalence transformation of the class~\eqref{eq:genKawahara}.
Substitution into equation~\eqref{eq:typeBsource} and division by
$\varepsilon$ give
\begin{equation}\label{eq:typeBtransformed}
 \tilde u_{\tilde t}
 +\frac{\ln(1+\varepsilon\tilde u)}{\varepsilon}
       \tilde u_{\tilde x}
 +\beta(\tilde t)\tilde u_{\tilde x\tilde x\tilde x}
 +\gamma(\tilde t)\tilde u_{\tilde x\tilde x\tilde x\tilde x\tilde x}=0.
\end{equation}
Since
\[
 \frac{\ln(1+\varepsilon\tilde u)}{\varepsilon}
 =\tilde u+O(\varepsilon),
\]
the limit $\varepsilon\to0$ of~\eqref{eq:typeBtransformed} is
\[
 \tilde u_{\tilde t}+\tilde u\tilde u_{\tilde x}
 +\beta(\tilde t)\tilde u_{\tilde x\tilde x\tilde x}
 +\gamma(\tilde t)\tilde u_{\tilde x\tilde x\tilde x\tilde x\tilde x}=0,
\]
which belongs to the class~\eqref{eq:genKawahara_u}. The family
\eqref{eq:typeB} becomes singular in the limit, and the limiting
nonlinearity satisfies $f_{uu}=0$. Hence no nonsingular admissible
transformation connects the source and target equations, even though
for every nonzero value of the contraction parameter the intermediate
transformation is an equivalence transformation within the nonlinear
subclass.

The maximal Lie invariance algebras contract under the same family.
After appropriate $\varepsilon$-dependent rescalings of the basis
operators, the two generators of the logarithmic kernel contract as
\begin{gather*}
 \varepsilon\partial_x
 \  \longrightarrow\ \partial_{\tilde x},\qquad
 \varepsilon(t\partial_x+u\partial_u)
 \ \longrightarrow\ \tilde t\partial_{\tilde x}+\partial_{\tilde u}.
\end{gather*}
Thus, the two-dimensional algebra of Case~1.3 contracts to the kernel
algebra of Table~\ref{Kawahara_tab:2}. The additional generators in
Cases~1.4 and~1.5 contract as
\begin{gather*}
 3(t\partial_t+x\partial_x)
  \ \longrightarrow\ 3\tilde t\partial_{\tilde t}
   +3\tilde x\partial_{\tilde x},\qquad
 \partial_t
 \ \longrightarrow\ \partial_{\tilde t}.
\end{gather*}
The first limiting operator is the additional generator of Case~2.1
with $\rho=2$, and the second is the additional generator of
Case~2.3. We therefore obtain the contractions
\[
 1.3\to2.0,\qquad 1.4\to2.1|_{\rho=2},\qquad 1.5\to2.3.
\]

The associated ans\"atze and reduced equations contract consistently.
For the kernel reduction, the source subalgebra
$\langle(t+a)\partial_x+u\partial_u\rangle$ has the ansatz
\[
 u=\exp\!\left(\frac{x}{t+a}\right)\varphi(t).
\]
Putting $x=\varepsilon\tilde x$ and
$\varphi=1+\varepsilon\tilde\varphi$ gives
\[
 \tilde u=\frac{\exp(\varepsilon\tilde x/(t+a))
 (1+\varepsilon\tilde\varphi)-1}{\varepsilon}
 \longrightarrow \frac{\tilde x}{t+a}+\tilde\varphi(t),
\]
which is the degenerate ansatz for the linear equation. The source
reduced equation is
\[
 \varphi'+\frac{\varphi\ln\varphi}{t+a}
 +\frac{\varepsilon^3\beta(t)\varphi}{(t+a)^3}
 +\frac{\varepsilon^5\gamma(t)\varphi}{(t+a)^5}=0.
\]
Substituting $\varphi=1+\varepsilon\tilde\varphi$, dividing by
$\varepsilon$ and passing to the limit give
\[
 (t+a)\tilde\varphi'+\tilde\varphi=0.
\]

For the contraction $1.4\to2.1|_{\rho=2}$, consider the reduction with
respect to $t\partial_t+x\partial_x$. The source invariant and ansatz
are $\omega=x/t$ and $u=\varphi(\omega)$. Write
\[
 \omega=\varepsilon\tilde\omega,\qquad
 \tilde\omega=\frac{\tilde x}{\tilde t},\qquad
 \varphi(\varepsilon\tilde\omega)
   =1+\varepsilon\tilde\varphi(\tilde\omega).
\]
For the family~\eqref{eq:typeBsource}, the source reduced equation is
\[
 \varepsilon^5\delta\varphi'''''
 +\varepsilon^3\lambda\varphi'''
 +(\ln\varphi-\omega)\varphi'=0.
\]
Since
$\varphi^{(k)}(\varepsilon\tilde\omega)
 =\varepsilon^{1-k}\tilde\varphi^{(k)}(\tilde\omega)$, division by
$\varepsilon$ and passage to the limit give
\[
 \delta\tilde\varphi'''''+\lambda\tilde\varphi'''
 +(\tilde\varphi-\tilde\omega)\tilde\varphi'=0,
\]
which is the reduced equation for Case~2.1 at $\rho=2$.

Finally, consider the contraction $1.5\to2.3$. In the source algebra
take the $\varepsilon$-dependent generator
\[
 a_\varepsilon\partial_t+t\partial_x+u\partial_u,
 \qquad a_\varepsilon=\frac{a}{2\varepsilon}.
\]
After multiplication by $\varepsilon$, this operator contracts as
\[
 \varepsilon\left(a_\varepsilon\partial_t
 +t\partial_x+u\partial_u\right)
 \ \longrightarrow\
 \frac12\bigl(a\partial_{\tilde t}
 +2\tilde t\partial_{\tilde x}+2\partial_{\tilde u}\bigr).
\]
The source ansatz is
\[
 u=e^{t/a_\varepsilon}\varphi(\omega),\qquad
 \omega=x-\frac{t^2}{2a_\varepsilon}.
\]
With
\[
 \omega=\varepsilon\tilde\omega,
 \qquad \tilde\omega=\tilde x-\frac{\tilde t^2}{a},
 \qquad
 \varphi(\varepsilon\tilde\omega)
   =1+\varepsilon\tilde\varphi(\tilde\omega),
\]
we obtain
\[
 \tilde u\longrightarrow
 \frac{2\tilde t}{a}+\tilde\varphi(\tilde\omega).
\]
The source reduced equation is
\[
 \varepsilon^5\delta\varphi'''''
 +\varepsilon^3\lambda\varphi'''
 +\ln\varphi\,\varphi'
 +\frac{1}{a_\varepsilon}\varphi=0.
\]
After division by $\varepsilon$, it contracts to
\[
 \delta\tilde\varphi'''''+\lambda\tilde\varphi'''
 +\tilde\varphi\tilde\varphi'+\frac{2}{a}=0,
\]
which is the reduced equation associated with
$\langle a\partial_t+2t\partial_x+2\partial_u\rangle$ in
Section~\ref{sec:reductions}.

The remaining linear symmetry-extension cases do not arise as
full-algebra contractions from logarithmic cases. For Case~2 and
Case~4 of Table~\ref{Kawahara_tab:2}, no logarithmic source case has
both the same maximal-algebra dimension and a compatible coefficient
pair. For Case~1, the dilation generator obtained from Case~1.4 agrees
with the target generator only for $\rho=2$. Thus, Case~1 with
$\rho\ne2$, Case~2 and Case~4 lie outside the present contraction
scheme. The $\nu=0$ representative of Case~4 is nonetheless the
carrier of the irreducible linear energy law~\eqref{eq:linenergyreps}.

The Type~B limits are also compatible with the conservation-law
classification, allowing $\varepsilon$-dependent linear combinations
and normalizations. The transformation rule for characteristics follows
from the factorization used above: since the left-hand side of the
source equation~\eqref{eq:typeBsource} equals $\varepsilon$ times the
left-hand side of the transformed equation~\eqref{eq:typeBtransformed}
and $D_x=\varepsilon^{-1}D_{\tilde x}$, a conserved current $(T,X)$ of
\eqref{eq:typeBsource} with the characteristic $Q$ gives rise to the
conserved current $(T,X/\varepsilon)$ of~\eqref{eq:typeBtransformed}
with the characteristic $\varepsilon Q$; the jet variables transform as
$u=1+\varepsilon\tilde u$ and
$\partial_x^ku=\varepsilon^{1-k}\partial_{\tilde x}^k\tilde u$, $k\geqslant1$.

Consider the contraction $1.4\to2.1|_{\rho=2}$. For the
family~\eqref{eq:typeBsource} with $\beta=\lambda t^2$,
$\gamma=\delta t^4$, the logarithmic law~\eqref{eq:logQT} holds with
$\lambda$ and $\delta$ replaced by $\varepsilon^3\lambda$ and
$\varepsilon^5\delta$. Its transformed characteristic expands as
\[
  \varepsilon Q
  =2\varepsilon^2\bigl(\tilde t\,\tilde u-\tilde x\bigr)
   +O(\varepsilon^3),
\]
the terms of order $\varepsilon$ cancelling identically and the
dispersive terms contributing only at order $\varepsilon^3$, since
$\varepsilon\cdot\varepsilon^5\delta t^5u_{xxxx}
=\varepsilon^3\delta\tilde t^{\,5}\tilde u_{\tilde x\tilde x\tilde x\tilde x}$
and similarly for the third-order term. Hence, after the normalization
of the conservation law by the factor $-1/(2\varepsilon^2)$, the
characteristic tends to
$\tilde x-\tilde t\,\tilde u$, the characteristic of the Galilean
law~\eqref{eq:linGal}. The density expands as
\[
  T=-\frac{\tilde t}{2}-\varepsilon\tilde x
    +\varepsilon^2\bigl(\tilde t\,\tilde u^2-2\tilde x\tilde u\bigr)
    +O(\varepsilon^3),
\]
where the $\tilde u$-independent summands correspond to the trivial
conservation laws with the conserved currents $(-\tilde t/2,\tilde x/2)$ and
$(-\varepsilon\tilde x,0)$; modulo these summands, the same
normalization $-1/(2\varepsilon^2)$ yields in the limit
$\tilde x\tilde u-\tfrac12\tilde t\,\tilde u^2$, the density
of~\eqref{eq:linGal}. Thus the nonautonomous logarithmic
law~\eqref{eq:logQT} of Case~1.4 contracts to the Galilean
law~\eqref{eq:linGal} of Case~$2.1|_{\rho=2}$.

For the contraction $1.5\to2.3$, the autonomous logarithmic
energy~\eqref{eq:CLenergyconst} with $F=u\ln u-u$ and
$H=\tfrac12u^2\ln u-\tfrac34u^2$, taken for the
family~\eqref{eq:typeBsource} with $\beta=\lambda$, $\gamma=\delta$,
degenerates in the naive limit: its transformed characteristic is
$\varepsilon Q=-\varepsilon+O(\varepsilon^3)$, with the leading term
proportional to the mass characteristic. A recombination with the
universal laws is therefore necessary. Adding the mass law,
$Q\mapsto Q+Q_1$ and $T\mapsto T+T_1$, gives the exact expansions
\[
  \varepsilon\bigl(Q+Q_1\bigr)
  =\varepsilon^3\Bigl(
    \delta\tilde u_{\tilde x\tilde x\tilde x\tilde x}
    +\lambda\tilde u_{\tilde x\tilde x}
    +\tfrac12\tilde u^2\Bigr)+O(\varepsilon^4),
  \qquad
  T+T_1=\frac14
  +\varepsilon^3\Bigl(
    \tfrac12\bigl(\delta\tilde u_{\tilde x\tilde x}^2
    -\lambda\tilde u_{\tilde x}^2\bigr)
    +\tfrac16\tilde u^3\Bigr)+O(\varepsilon^4),
\]
so that, after removing the trivial constant summand and normalizing
the conservation law by $\varepsilon^{-3}$, the limit is exactly the
autonomous cubic energy of Case~2.3, that is,
\eqref{eq:CLenergyconst} for $f=u$ with $H=u^3/6$.

The universal laws~\eqref{eq:CLuniversal} require no limits at all:
$\varepsilon Q_1/\varepsilon=1$ and
$(T_1-1)/\varepsilon=\tilde u$ for the mass law, and the recombination
$Q_2-Q_1$ gives, exactly,
\[
  \frac{\varepsilon(Q_2-Q_1)}{\varepsilon^2}=\tilde u,
  \qquad
  \frac{T_2-T_1+\tfrac12}{\varepsilon^2}
  =\frac{(u-1)^2}{2\varepsilon^2}=\frac{\tilde u^2}{2}.
\]
Hence the space spanned by the universal laws~\eqref{eq:CLuniversal} is
preserved along every Type~B family. For the generic arrow $1.3\to2.0$,
however, this span is, by the classification of
Section~\ref{sec:CL}, the entire space of low-order conservation laws
of the source equations, and the transformed characteristic of an
arbitrary linear combination,
$\varepsilon(c_1Q_1+c_2Q_2)
=\varepsilon(c_1+c_2)+\varepsilon^2c_2\tilde u$,
has, for any $\varepsilon$-dependent coefficients and normalizations,
only limits lying in the span of $1$ and~$\tilde u$. The Galilean
characteristic $\tilde x-\tilde t\,\tilde u$ depends explicitly
on~$\tilde x$ and lies outside this span, so no additional logarithmic
source law exists that contracts to the target Galilean law.

\section{Conclusion}

In this paper the extended group analysis of the class~\eqref{eq_ggKawahara} of generalized Kawahara equations with time-dependent coefficients is carried out. The starting point is the description of the equivalence groupoid of this class. Since the class is not normalized, we partition it into two subclasses singled out by the conditions $f_{uu}\neq0$ and $f_{uu}=0$. Each of these subclasses is normalized in the extended generalized sense, and the corresponding equivalence groups $\hat G^{\sim}$ and $\hat G_0^{\sim}$ are constructed in Theorems~\ref{theorem_fuune0} and~\ref{theorem_fuu=0}. Knowledge of the structure of these groups allows us to choose the optimal gaugings of the arbitrary elements, $\alpha=1$ for the subclass with $f_{uu}\neq0$ and $\alpha=1$, $\kappa_0=0$ for the subclass with $f_{uu}=0$. These gaugings reduce the group classification problems for the two subclasses to those for the classes~\eqref{eq:genKawahara} and~\eqref{eq:genKawahara_u}, which are normalized in the usual sense. As a byproduct of the study of admissible transformations, we derive an easy-to-verify criterion of reducibility of a variable coefficient equation from the class~\eqref{eq_ggKawahara} to a constant coefficient one.

The results of the group classification are summarized in Tables~\ref{Kawahara_tab:1} and~\ref{Kawahara_tab:2}. For the equations with $f_{uu}\neq0$ the kernel of the maximal Lie invariance algebras is the one-dimensional algebra $\langle\partial_x\rangle$ and there exist exactly nine inequivalent cases of Lie symmetry extension, whereas for the equations with $f_{uu}=0$ the kernel is the two-dimensional algebra $\langle\partial_x,\,T(t)\partial_x+\partial_u\rangle$, $T_t=\alpha$ (represented by $\langle\partial_x,\,t\partial_x+\partial_u\rangle$ in the gauge $\alpha=1$) and there are four inequivalent cases of extension. In all the cases the maximal Lie invariance algebras are at most three-dimensional, and, in contrast to the earlier conference presentations~\cite{vane2020c,VMZh2023}, we identify their isomorphism types within the classification of low-dimensional Lie algebras. Since fitting a specific model to a gauged canonical form may be inconvenient in applications, we also derive, using the equivalence-based approach, the classification lists in which the coefficients are not simplified by point transformations (see Tables~\ref{Kawahara_tab:3} and~\ref{Kawahara_tab:4}). In particular, these tables involve the nonlocality $T=\int\alpha(t)\,{\rm d}t$ and additional constant parameters that are invisible at the level of canonical forms but unavoidable when symmetries of a given physical model are written out explicitly.
Then the exhaustive classification of Lie reductions with respect to the optimal systems of one-dimensional subalgebras is presented in a unified form for both subclasses, and the list of known exact solutions of variable coefficient Kawahara equations is extended with new ones constructed in closed form.

The classification of low-order conservation laws shows that every equation from the class~\eqref{eq_ggKawahara} admits conservation laws of mass and of the squared $L^2$ norm, while energy-type conservation laws exist only for distinguished coefficient branches, which agree with the cases singled out by the symmetry classification. The conservation laws obtained for the reciprocal-power, logarithmic and exponential nonlinearities supplement the classification of~\cite{GandariasRosaRecioAnco}, and the linear energy branch includes the representative with $\gamma=(t^2+1)^{3/2}$ that is not reducible to power-type coefficients.

We establish that the cases of Lie symmetry extension are linked by contractions of two types. Contractions of Type~A connect the power nonlinearity cases to the exponential ones within the subclass with $f_{uu}\neq0$ and are induced by a family of equivalence transformations. Contractions of Type~B connect the logarithmic nonlinearity cases to the linear ones and are genuine singular limits, since the two subclasses are disjoint and are not related by point transformations. In both types the equations, the maximal Lie invariance algebras, the ans\"atze and the reduced equations are shown to transform consistently in the limit. The conservation-law classification is compatible with these limits as well. On the distinguished coefficient branch~\eqref{eq:powerbranch}, along which the exponent $\rho$ varies with $n$, the Type~A limit maps the nonautonomous power-law conservation law to the exponential one. For Type~B, the nonautonomous logarithmic law of Case~1.4 contracts to the Galilean law of Case~$2.1|_{\rho=2}$, and the autonomous logarithmic energy of Case~1.5 contracts to the autonomous cubic energy of Case~2.3 after an appropriate recombination with the universal laws. The universal-law space is preserved along all Type~B families, whereas the generic contraction $1.3\to2.0$ has no additional logarithmic source law corresponding to the target Galilean law.

Beyond the scope of the present work, several directions remain open for future investigation. These include the classification of generalized and potential symmetries, as well as higher-order conservation laws, for the variable-coefficient case; the study of Kawahara-type models with coefficients depending also on the spatial variable; and applications of the derived reductions to invariant boundary-value problems. The constructed conservation laws may also be used to introduce potential systems and to test the accuracy of numerical simulations of the corresponding wave models.

\subsection*{Acknowledgements}
The authors would like to thank Prof. Roman Popovych for useful discussions.
This work was supported by a grant from the Simons Foundation (SFI-PD-Ukraine-00014586, O.V., A.Z.).


\end{document}